\documentclass[letter,reqno,11pt]{article}
\usepackage[pdfstartview=FitH,pdfpagemode=None,colorlinks=true,citecolor=verydarkblue,linkcolor=verydarkblue]{hyperref}
\usepackage{amssymb,amsmath,graphicx,verbatim,amsthm}
\usepackage{color,enumerate}
\usepackage[hmargin=1.2in,vmargin=1.2in]{geometry} 
\newcommand{\new}[1]{{\textcolor{black}{#1}}} 
\newcommand{\iddo}[1]{{\textcolor{blue}{#1}}} 
\newcommand{\iddosml}[1]{{\textcolor{black}{#1}}} 
\newcommand{\iddoacc}[1]{{\textcolor{black}{#1}}} 
\newcommand{\iddodefer}[1]{} 
\newcommand{\Commentdel}[1]{}

\newtheorem{theorem}{Theorem}
\newtheorem{lemma}[theorem]{Lemma}
\newtheorem{claim}[theorem]{Claim}
\newtheorem{definition}{Definition}
\newtheorem{proposition}[theorem]{Proposition}
\newtheorem{corollary}[theorem]{Corollary}
\newtheorem{remark}[theorem]{Remark}

\newtheorem*{claim*}{Claim}
\newtheorem*{claim1}{Claim 1}
\newtheorem*{claim2}{Claim  2}
\newtheorem*{fact*}{Fact}

\newcounter{thm1}

\def\squareforqed{\hbox{\rlap{$\sqcap$}$\sqcup$}}
\def\qed{\ifmmode\squareforqed\else{\unskip\nobreak\hfil
\penalty50\hskip1em\null\nobreak\hfil{\tt QED}
\parfillskip=0pt\finalhyphendemerits=0\endgraf}\fi}
\newcommand{\wh}[1]{\widehat{#1}}

\newcommand\ring{{\F}}

\newcommand\F{{\mathbb F}}
\newcommand\Q{{\mathbb Q}}
\newcommand{\ignore}[1]{}
\newcommand\PP{{\mathbb P}}
\newcommand\PF{{\PP_f}}
\newcommand\PC{\ensuremath{\PP_c}}
\newcommand\PI{\PP_{c}^{{-1}}}

\newcommand\mez{~~~~~}
\newcommand\cplus{{\oplus}}
\newcommand\ctimes{{\otimes}}
\newcommand{\tr}[1]{\left[#1\right]}

\newcommand\Num{{\hbox{Num} }}
\newcommand\Den {{\hbox{Den} }}

\newcommand\Inv{\ensuremath{\hbox{\rm Inv}}}
\newcommand\inv{{\hbox{\rm pow}}}

\newcommand\poly{\hbox{{\rm poly}}}
\newcommand\Det{\hbox{{\rm DET}}}
\newcommand\detc{{\det}_c}
\newcommand\Adj{\hbox{{\rm Adj}}}
\newcommand\Coef{\Delta}
\newcommand\dd{{\delta}}

\newcommand {\cd}{\cdot}
\newcommand{\Base}{\mbox{}\\ \ind{\textit{Base case: }}}
\newcommand{\Induction}{\mbox{}\\ \ind{\textit{Induction step: }}}
\newcommand{\case}[1]{\ind\textbf{Case #1}:\,}
\newcommand{\induction}{\Induction}

\newcommand{\mardel}[1]{}
\newcommand {\ind} {\noindent}

\newcommand {\para}[1] {\paragraph{#1}}

\DeclareMathAlphabet{\mathitbf}{OML}{cmm}{b}{it}

\newcommand{\mbf}[1]{\ensuremath{\mathbf{#1}}}
\newcommand{\NCT}{\ensuremath{\mathbf{NC}^2}}

\newcommand{\NP}{\mbf{NP}}

\font\sf=cmss10
\newcommand{\Nats}{{\hbox{\sf I\kern-.13em\hbox{N}}}}   
\newcommand{\Reals}{{\hbox{\sf I\kern-.14em\hbox{R}}}}  
\newcommand{\Ints}{{\hbox{\sf Z\kern-.43emZ}}}          
\newcommand{\CC}{{\hbox{\sf C\kern -.48emC}}}           
\newcommand{\QQ}{{\hbox{\sf C\kern -.48emQ}}}           

\ifx\theorem\undefined
\newtheorem{theorem}{Theorem}[section]
\newtheorem{lemma}[theorem]{Lemma}
\newtheorem{proposition}[theorem]{Proposition}
\newtheorem{corollary}[theorem]{Corollary}

\newenvironment{remark}{\HalfSpace\par\noindent{\bf Remark}:}{\HalfSpace}
\fi
\newenvironment{notation}{\QuadSpace\par\noindent{\bf Notation}:}{\HalfSpace}

\ifx\begin{proof}\undefined
\newenvironment{proof}{\QuadSpace\par\noindent{\bf Proof}:}{\end{proof}\HalfSpace}
\fi

\newcommand{\QuadSpace}{\vspace{0.25\baselineskip}}
\newcommand{\HalfSpace}{\vspace{0.5\baselineskip}}
\newcommand{\FullSpace}{\vspace{1.0\baselineskip}}

\definecolor{bluetxt}{rgb}{0,0,.5}
\definecolor{myred}{rgb}{0.6,0.0,0.1}
\definecolor{greentxt}{rgb}{0,.32,0}
\definecolor{redtxt}{rgb}{0.1,0.1,0.65}
\definecolor{purpletxt}{rgb}{0.6,0.1,0.7}
\definecolor{black}{rgb}{.0,.0,.0}
\definecolor{verydarkblue}{rgb}{.0,.0,.4}
\definecolor{lightgray}{rgb}{.7,.7,.7}

\ifx\begin{proof}\undefined
\newenvironment{proof}{

\smallskip
\noindent\emph{Proof.}}{\hfill\(\Box\)
\bigskip
}
\fi
\newlength{\defbaselineskip}
\newcommand{\doublespacing}{\setlength{\baselineskip}{1.0\defbaselineskip}}
\begin{document}

\date{April 2013}
\title{Short Proofs for the Determinant Identities\footnote{Conference version appeared in STOC 2012.}}
\author{Pavel Hrube\v{s}\thanks{Computer Science and Engineering, University of Washington. Email: \texttt{pahrubes@gmail.com}} \and Iddo Tzameret\thanks{Institute for Theoretical Computer Science, The Institute for Interdisciplinary Information Sciences (IIIS), Tsinghua University, Beijing, 100084, China. Email: \texttt{tzameret@tsinghua.edu.cn}.~Supported in part by the National Basic Research Program of China Grant 2011CBA00300, 2011CBA00301, the National Natural Science Foundation of China Grant 61033001, 61061130540, 61073174, 61150110582.}}
\maketitle

\doublespacing 
\begin{abstract}
\doublespacing 
We study arithmetic proof systems $\PC(\F)$ and $ \PF(\F)$ operating with arithmetic circuits and arithmetic formulas, respectively, that prove polynomial identities over a field $\F$. We establish a series of structural theorems about these proof systems, the main one stating that $\PC(\F)$ proofs can be balanced: if a polynomial identity of syntactic degree $ d $ and depth $k$ has a  $\PC(\F)$ proof of size $s$,  then it also has a $\PC(\F)$ proof of size $ {\rm poly}(s,d) $ and depth $ O(k+\log^2 d + \log d\cd \log s) $. As a corollary, we obtain a quasipolynomial simulation of $\PC(\F)$ by $\PF(\F)$, for identities of a polynomial syntactic degree.

Using these results we obtain the following: consider the identities
\begin{eqnarray*}
 \det(XY) = \det(X)\cd\det(Y) \quad\mbox{ and }\quad  \det(Z)= z_{11}\cdots z_{nn},
\end{eqnarray*}
where  $X,Y $ and $ Z$ are $n\times n$ square matrices and $Z$ is a triangular matrix with $z_{11},\dots, z_{nn}$ on the diagonal (and $ \det $ is the determinant polynomial). Then we can construct a polynomial-size arithmetic circuit $\det$ such that the above identities have $ \PC(\F)$ proofs of polynomial-size and $ O(\log^2 n)$ depth.  Moreover, there exists an arithmetic formula $ \det $ of size $n^{O(\log n)}$ such that the above identities have $\PF(\F)$ proofs of size $n^{O(\log n)}$.

This yields a solution to a \new{basic} open problem in propositional proof complexity, namely, whether there are polynomial-size \NCT-Frege proofs for the determinant identities and the \emph{hard matrix identities}, as considered, e.g.~in Soltys and Cook \cite{SC04} (cf., Beame and Pitassi \cite{BP98}). We show that matrix identities like $ AB=I \rightarrow BA=I $ (for matrices over the two element field) as well as basic properties of the determinant have polynomial-size \NCT-Frege proofs, and quasipolynomial-size Frege proofs.
\end{abstract}

\section{Introduction}
The field of proof complexity is dominated by the question of how hard is it to prove propositional tautologies. For weak proof systems, such as resolution, many hardness results are known (cf., \cite{Seg_BSL07} for a recent technical survey), but for strong propositional proof systems like Frege or extended Frege the question remains completely open. In this paper we continue to investigate a different but related problem:  how hard is it to prove polynomial identities? For this purpose, various systems for proving polynomial identities were introduced in \cite{HT08}. The main feature of these systems is that they manipulate arithmetic equations of the form $F=G$, where $F,G$ are arithmetic formulas over a given field. Such equations are manipulated by means of simple syntactic rules, in such a way that $F=G$ has a proof if and only if $F$ and $G$ compute the same polynomial. The central question in this framework is the following:
\begin{quote}
What is the \emph{length} of such proofs, namely, does every true polynomial equation have a short proof, or are there hard equations that require extremely long proofs?
\end{quote}
In this paper, we focus on two arithmetic equational proof systems (arithmetic proofs systems, for short) for proving polynomial identities: $ \PF$ and $ \PC$. The former system was introduced in \cite{HT08} and the latter is an extension of it. The difference between the two systems is that $ \PF$ operates with \emph{arithmetic formulas}, whereas $\PC $ operates with \emph{arithmetic circuits}---this is analogous to the distinction between Frege and extended Frege proof systems (Frege and extended Frege proofs are propositional proof systems establishing propositional tautologies, essentially operating with boolean formulas and circuits, respectively).

The study of proofs of polynomial identities is motivated by at least two reasons. First, as a study of the Polynomial Identity Testing (PIT) problem. As a decision problem, polynomial identity testing can be solved by an efficient randomized algorithm \cite{Sch80,Zip79}, but no efficient deterministic algorithm is known. In fact, it is not even known whether there is a polynomial time non-deterministic algorithm or, equivalently, whether  PIT is in \NP. A proof system such as $\PC $ can be interpreted as a specific non-deterministic algorithm for PIT: in order to verify that an arithmetic formula $F$ computes the zero polynomial, it is sufficient to guess a proof of $F=0$ in $\PC$. Hence, if every true equality has a polynomial-size proof then PIT is in \NP. Conversely, $ \PF $ and $ \PC $ systems capture the common syntactic procedures used to establish equality of algebraic expressions. Thus, showing the existence of identities that require superpolynomial arithmetic proofs would imply that those syntactic procedures are not enough to solve PIT efficiently.

The second motivation comes from propositional proof complexity. The systems $\PF$ and $\PC$ are in fact \emph{restricted} versions of their propositional counterparts, Frege and extended Frege, respectively (when operating over $ GF(2) $). One may hope that the study of the former would help to understand the latter. Arithmetic proof systems have the advantage that they work with arithmetic circuits. The structure of arithmetic circuits is perhaps better understood than the structure of their Boolean counterparts, or is at least different, \new{suggesting} different techniques and fresh perspectives.

In order to understand the strength of the systems $\PF$ and $\PC$, as well as their relative strength, we investigate quite a specific question, namely, how hard is it to prove basic properties of the determinant? In other words, we investigate lengths of proofs of identities such as $\,\det(A B)=\det(A)\cd\det(B)$, or the cofactor expansion of the determinant. We show that such identities have polynomial-size $\PC$ proofs of depth $ O(\log^2 n) $ and quasipolynomial size $\PF$ proofs (both results hold over any field).\footnote{The parameter $ n $ is the dimension of the matrices $ A,B $, and quasipolynomial size means size $n^{O(\log n)}$.}

The determinant polynomial has a central role in both linear algebra and arithmetic circuit complexity. Therefore, an  immediate motivation for our inquiry is to understand whether arithmetic proof systems are strong enough to reason efficiently about the determinant. More importantly, we take the determinant question as a pretext to present several structural properties of $ \PC$ and $\PF$. A large part of this work is not concerned with the determinant at all, but is rather a series of general theorems showing how classical results in arithmetic circuit complexity can be translated to the setting of arithmetic proofs. We thus show how to capture efficiently the following results: (i) homogenization of arithmetic circuits (implicit in \cite{Str73}); (ii) Strassen's technique for eliminating division gates over large enough fields (also in \cite{Str73}); (iii) eliminating division gates over small fields---this is done by simulating large fields in small ones; and (iv) balancing arithmetic circuits (Valiant et al.~\cite{VSB+83}; see also \cite{Hya79}). Most notably, the latter result gives a collapse of polynomial-size $ \PC $ proofs to polynomial-size $ O(\log^2 n) $-depth  $ \PC $ proofs (for proving identities of polynomial syntactic degrees) and a quasipolynomial simulation of $\PC$ by $\PF$. This is one important point where the arithmetic systems differ from Frege and extended Frege, for which no non-trivial simulation is known.

Furthermore, the proof complexity of linear algebra attracted a lot of attention in the past.
This was motivated, in part, by the goal of separating the propositional proof systems Frege and extended Frege.
A classical example, originally proposed by Cook \new{and Rackoff} (cf.,  \cite{BP98,SC04,SU04,Sol_PhD,Sol05}), is the so called \emph{inversion principle} asserting that  $ AB=I \rightarrow BA=I $. When $A,B$ are $n\times n$ matrices over $GF(2)$, the inversion principle is a collection of propositional tautologies. Soltys and Cook \cite{SC04,Sol_PhD} showed that the principle has polynomial size extended Frege proofs. On the other hand, no feasible Frege proof is known, and hence the inversion principle is a candidate for separating the two proof systems. Other candidates, including several based on linear algebra, were presented by Buss et al.~\cite{BBP95}.  The inversion principle is one of the ``hard matrix identities'' explored in \cite{SC04}. Inside Frege, the hard matrix identities have feasible proofs from one another, and they have short proofs from the aforementioned determinant identities. This connection between the hard matrix identities and the determinant identities serves as an evidence for the conjecture that hard matrix identities require superpolynomial Frege proofs: it seems that every Frege proof must in some sense construct  the determinant, which is believed to require a superpolynomial-size formula.

A related question is whether the hard matrix identities and the determinant identities have polynomial-size \NCT-Frege proofs\footnote{That is, polynomial size proofs using circuits of $O(\log^{2}n)$-depth}. This was conjectured in, e.g.,~\cite{BBP95}, based on the intuition that the determinant is \NCT\ computable, 
and so by the analogy between circuit classes and proofs, it is natural to assume that the determinant properties are efficiently provable in \NCT-Frege.  Again,  a polynomial-size extended Frege proofs of the determinant identities have been constructed in \cite{SC04}. Whether these identities have polynomial-size \NCT-Frege proofs (and hence, quasipolynomial-size Frege proofs) remained open.
In this paper, we positively answer this question: we show that over $GF(2)$, the hard hard matrix identities and the determinant identities have polynomial-size $\NCT$-Frege proofs.  This is  a simple corollary of the results on arithmetic proof systems. Over the two element field, an $O(\log^{2}n)$-depth $ \PC $ proof is formally also \NCT-Frege proof\footnote{When $ + $ and $ \cd $ modulo $ 2$ are interpreted as Boolean connectives and $ = $ is interpreted as logical equivalence.}. Thus, if determinant identities like $\det(AB)=\det(A)\cd\det(B)$ have polynomial-size $\PC(GF(2))$ proofs with depth $ O(\log^2 n) $, then the corresponding propositional tautologies have polynomial-size \NCT-Frege proofs.

Let us remark that one can also consider propositional translations of the determinant identities (and the hard matrix identities) over different finite fields or even the rationals. We do not explicitly study these translations, but there is no apparent obstacle to extending the result to these cases.


To understand our construction of short arithmetic proofs for the determinant identities, let us consider the following example. In \cite{Ber84}, Berkowitz constructed a quasipolynomial size arithmetic formula for the determinant. He used a clever combinatorial argument designed \emph{specifically} for the determinant function. However, one can build such a formula in a completely oblivious way: first compute the determinant by, say, Gaussian elimination algorithm. This gives an arithmetic circuit with division gates. Second, show that any circuit with division gates computing a polynomial can be efficiently simulated by a division-free circuit \cite{Str73}, and finally, show that any arithmetic circuit of a polynomial degree can be transformed to an $ O(\log^2 n) $-depth circuit computing the same polynomial, with only a polynomial increase in size \cite{VSB+83} (or to a formula with at most a quasipolynomial increase in size \cite{Hya79}). This paper follows a similar strategy, but in the proof-theoretic framework. 

It should be stressed that in full generality, the structural theorems about $\PC$ and $\PF$ \emph{cannot} be reproduced for propositional Frege and extended Frege systems.  As already mentioned, no non-trivial simulation between Frege and extended Frege is known, and the other theorems are difficult to even formulate in the Boolean context.  This also illustrates one final point: in order to construct a Frege proof of a tautology $T$, it may be useful to interpret $T$ as a polynomial identity and prove it in some of the---weaker but better structured---arithmetic proof systems.

\subsection{Arithmetic proofs with circuits and formulas}\label{sec:equational proofs}
Before presenting and explaining the main results of this paper (in Section \ref{sec:results}), we need to introduce our basic arithmetic proof systems.

\paragraph{Arithmetic circuits and formulas.}
Let $\ring$ be a field.  \emph{An arithmetic circuit} $F$ is a finite directed acyclic graph as follows.
Nodes (or gates) of in-degree zero are labeled by either a variable
or a field element in $\F$.
All the other nodes have in-degree two and they are labeled by
either $+$ or $\times$. Unless stated otherwise, we assume that $F$ has exactly one node of out-degree zero, called the \emph{output node}, and that moreover the
 two edges going into a gate $v$ labeled by $\times$ or $+$ are
labeled by \emph{left} and \emph{right}.  This
is to determine the order of addition and multiplication\footnote{Although ultimately, addition and multiplication are commutative. }.
An arithmetic circuit is called a \emph{formula}, if the out-degree of each node in it is one (and so the underlying graph is a directed tree). The \emph{size} of a circuit is the number of nodes in it,
and the \emph{depth} of a circuit is the length of the longest directed path in it. Arithmetic circuits and formulas will be referred to simply as \emph{circuits} and \emph{formulas}.
\smallskip

For a circuit $ F $ and a node $u$ in $ F $, $ F_u$ denotes the subcircuit of $F$ with output node $u$. If $F,G$ are circuits then
\[
\mbox{  $F\cplus G$ and $F\ctimes G$ }
\]
abbreviate any circuit $H$ whose output node is $u+v$ and $ u\cdot v$, respectively, where $H_u=F$ and $H_v=G$. Furthermore, \[
\mbox{  $F+ G$ and $F\cdot G$ }
\]
 denote the unique circuit \new{of the form $F'\cplus G'$ and $F'\ctimes G'$, respectively, where $F'$, $G'$ are disjoint copies of $F$ and $G$. In particular, if $F$ and $G$ are formulas then so are $F+ G$ and $F\cdot G$.}

A circuit $F$ computes a polynomial $\widehat {F}$ with coefficients from $\ring$ in the obvious manner.
That is, if $F$ consists of a single node labeled with $z$,  a variable or an element of $\F$, we have $\widehat F:=z$. Otherwise, $F$ is either of the form $G\cplus H$ or $G\ctimes H$, and we let $\widehat F:=\widehat G+\widehat H$ or $\widehat F:=\widehat G\cdot \widehat H$, respectively.

\new{
\emph{Substitution} is understood in the following sense. Let $F=F(z)$ be a circuit and $z$ a variable. For a circuit $G$, the circuit $F(G)$ is defined as
follows: let $z_{1},\dots, z_{k}$ be the nodes in $F$ labeled by $z$. Introduce $k$ disjoint copies $G_{1},\dots, G_{k}$ of $G$, and let $F(G)$ be the
union of $F,G_{1},\dots, G_{k}$ where we replace the node $z_{i}$ by the output node of $G_{i}$. Specifically, if $F$ and $G$ are formulas then so is $F(G)$. The circuit $F(G)$
will also be written as $F(z/G)$.
 }

\subsubsection*{The system $\PF(\ring)$}
We now define two proof systems for deriving polynomial identities. The systems manipulate \emph{arithmetic equations}, that is, expressions of the form $F=G$. In the case of $\PF(\ring)$,  $F,G$ are formulas, and in the case of $ \PC(\ring)$, $ F,G $ are circuits (see \cite{HT08} for similar proof systems).

Let $\ring$ be a field. The system $\PF(\ring)$ proves equations of the form $F=G$, where $F,G$ are \emph{formulas} over $\F$.
The inference rules are:
\begin{align*}
&{\rm R1}&   \frac{F=G}{G=F} \mez\mez\mez &\qquad\qquad {\rm R2}& \mez \frac{F=G\qquad G=H}{F=H} && \\
&{\rm R3}&  \frac{F_1=G_1 \qquad  F_2=G_2}{F_1+ F_2= G_1+ G_2}\mez& \qquad\qquad{\rm R4}& \mez\frac{F_1=G_1\qquad F_2=G_2}{F_1\cdot F_2= G_1\cdot G_2} .&&
\end{align*}
The axioms are equations of the following form, with $F,G,H$ formulas:

\begin{tabular}{l  l    l} A1&\mez  $F=F$  & \\
         A2    &\mez $F+G=G+F $                  &    A3    \mez   $F+(G+H) = (F+G)+H$\\
         A4    &\mez $F\cdot G=G\cdot F $, &   A5    \mez $ F \cdot (G\cdot H)= (F\cdot G)\cdot H$ \\
         A6    &\mez $F\cdot (G+H) = F\cdot G+F\cdot H$ &   \\
         A7    &\mez $F+0 = F$                      &   A8    \mez  $F\cdot 0 = 0$\\
         \iddosml{A9}    & \iddosml{\mez $F\cd 1 = F$}   &   \\
        A10    &\mez $a=b+c\,,     a^\prime=b^\prime\cdot c^\prime\,, $& if  $a,b,c,a^\prime,b^\prime,c^\prime\in  \ring$, are such that\\
        & &   the equations  hold in
$\ring$.
\end{tabular}
\smallskip

The rules and axioms can be divided into two groups. The rules
 R1-R4 and axiom A1 determine the logical properties of equality ``='', and axioms A2-A10 assert that  polynomials  form a commutative ring over $\F$.

\emph{A proof $S$ in $\PF(\ring)$} is a sequence of equations $F_1=G_1,\, F_2=G_2,\dots, F_k=G_k$, with $F_i, G_i$ formulas, such that every equation is either an axiom A1-A10, or was obtained from previous equations by one of the rules R1-R4. An equation $ F_i=G_i $ appearing in a proof is also called a \emph{proof line}. We consider two measures of complexity for $S$:  the \new{\emph{size of $S$}}
 is the sum of the sizes of $F_i$ and $G_i$, $i\in [k]$, and the \emph{number of proof lines in $S$} is $ k $. (Throughout the paper, $[k]$ stands for $\{1,\dots,k\}$.)

\subsubsection*{The system $\PC(\ring)$}
The system $\PC(\ring)$ differs from $\PF(\ring)$ in that it manipulates equations  with \emph{circuits}.
$\PC(\ring)$ has the same rules R1-R4  and axioms A1-A10 as $\PF(\F)$, but with $F,G,H, F_1,F_2,G_1,G_2 $ \emph{ranging over circuits}, augmented with the following two axioms:
\QuadSpace

\begin{tabular}{l  l    l}
C1&\mez  $F_1\cplus F_2=F_1+F_2$ \qquad \qquad\qquad &C2\mez $F_1\ctimes F_2=F_1\cdot F_2$.
\end{tabular}
\QuadSpace

A \emph{proof in $\PC(\F)$} is a sequence of equations  $F_1=G_1,\dots, F_k=G_k$, where $F_i, G_i$ are circuits,  and every equation is either an axiom or was derived by one of the rules. As for $\PF(\F)$, the \emph{size} of a proof is the sum of the sizes of all the circuits $F_i$ and $G_i$, $ i\in[k] $, and the \emph{number of proof lines} of the proof is $ k $. The \emph{depth} of a $\PC(\F)$ proof is the maximal depth of a circuit appearing in the proof.
\FullSpace

The main property of the two proof systems $ \PC(\F) $ and $ \PF(\F) $ is that they are sound and complete with respect to polynomial identities.  The systems prove an equation $F=G$ if and only if $F, G$ compute the same polynomial:
\begin{proposition} Let $\F$ be a field.
\begin{enumerate}[(i)]
\item \label{first: 1} For any pair $ F,G$ of arithmetic \new{formulas}, $\PF(\ring)$ proves $F=G$ iff $\widehat F= \widehat G$.
\item \label{first: 2} For any pair $ F,G$ of arithmetic circuits, $\PC(\ring)$ proves $F=G$ iff $\widehat F= \widehat G$.
\end{enumerate}
\end{proposition}
\new{Part \ref{first: 1} was shown in \cite{HT08}, part \ref{first: 2} is almost identical. Soundness can be easily proved by induction on the number of lines and completeness by rewriting $F$ and $G$ as a sum of monomials. }

\new{It should be noted that $\PF$ and $\PC$ proofs are closed under substitution. If $F_{1}=G_{1},\dots, F_{k}= G_{k}$ is a $\PC$ proof, $z$ a variable and $H$ a circuit then $F_{1}(z/H)=G_{1}(z/H),\dots, F_{k}(z/H)= G_{k}(z/H)$ is also a $\PC$ proof (similarly for $\PF$ and a formula $H$).
This means that from a general proof,  one can obtain the proof of its instance. }

For simplicity, we often suppress the explicit dependence on the field $\F $ in $\PC$ and $ \PF$, if the relevant statement holds over any field.

\para{Comments on the proof systems.}
The system $\PC$ is an algebraic analogue of the propositional proof system \emph{circuit Frege} (CF). Circuit Frege is polynomially equivalent to the more well-known \emph{extended Frege} system (EF) (see \cite{Kra95,Jer04}). Following this analogy, one can define  an \emph{extended $\PF$} proof system, $\hbox{\rm E}\PF$, as follows: \new{an $\hbox{\rm E}\PF$ proof is a $\PF$ proof in which  we are allowed to introduce new ``extension'' variables $z_{1}, z_{2},\dots$~via the axiom $z_{i}= F_{i}$, where we require that (i) the variable $ z_i $ appears in neither $ F_{i} $ nor in any previous proof-line; and (ii) the last equation in the proof contains none of the extension variables $ z_1,z_2,\ldots $ .}

The following is completely analogous to the propositional case (see \cite{Kra95,Jer04}):
\begin{proposition}\label{prop:EF}
\mbox{}
\begin{enumerate}[(i)]
\item The systems $\PC$ and $\hbox{{\rm E}}\PF$ polynomially simulate each other. More exactly, there is a polynomial $p$ such that for every pair of formulas $F,G$, if $F=G$ has a $\PC$ proof of size $s$ then it has an $\hbox{{\rm E}}\PF$ proof of size $p(s)$, and if $F=G$ has an $\hbox{{\rm E}}\PF$ proof of size $s$ then it has a $\PC$ proof of size $p(s)$.
\item If $F $ and $G$ are circuits of size $s$ and $F=G$ has a $\PC$ proof with $k$ proof lines then $F=G$ has a $\PC$ proof of
size $\poly(s,k)$.
\end{enumerate}
\end{proposition}

The second part of this statement is especially useful, because it is often easier to estimate the number of lines in a proof rather than its size.

\begin{remark}\label{def:similar}
\new{An alternative, polynomially equivalent, definition for $ \PC$ can be given as follows.
For a circuit $F$, define $F^{\bullet}$ as the unfolding of $F$ into a \emph{formula}. That is, $F^{\bullet}: = F$, if $F$ is a leaf, and
$(G\cplus H)^{\bullet}:= G^{\bullet}+ H^{\bullet}$, ${(G\ctimes H)}^{\bullet}:= G^{\bullet}\cdot H^{\bullet}$. We say that $ F$ and $G$ are \emph{similar circuits}, if $F^{\bullet}$ is the same formula as $G^{\bullet}$.
Then {\rm A1, C1, C2} could be replaced by the following single axiom: }
\vspace{-5pt}
\[
        \mbox{{\rm A1'} \hspace{50pt}$F=G$, \mez whenever $F$ and $G$ are similar.}
\]
The axiom {\rm A1'}  can be proved from {\rm A1, C1, C2} by a polynomial-size proof, and vice versa.
\end{remark}

\para{Notation for matrices inside proofs.} \label{matrix notation}
In this paper, matrices are understood as matrices whose entries are circuits and operations on matrices are operations on circuits. \new{We illustrate this for square matrices. Let ${F}=\{F_{{ij}}\}_{i,j\in [n]} $ be an $n\times n$ matrix whose entries are circuits $F_{ij}$; and similarly
${G}= \{G_{ij}\}_{i,j\in [n]}$. Addition and multiplication is defined in the obvious way, namely}
\[F+G=\{{F}_{ij}+ {G}_{ij}\}_{i,j\in [n]}\,,\,\, F\cdot G=\left\{ \sum\nolimits_{p=1}^{n}{F}_{ip}\cdot {G}_{pj}\right\}_{i,j\in [n]}\,,\]
where $+$ and $\cdot$ on the right-hand side is addition and multiplication on circuits.
If $a$ is a \new{single circuit},   $ a\cdot F $ is the matrix $ \{a \cd F_{ij}\}_{i,j\in[n]}$.
An equation ${F}={G}$ denotes the set of equations ${F}_{ij}={G}_{ij},\, i,j\in [n]$.

\section{Overview of results and techniques}\label{sec:results}

\subsection{Main theorem}
It is well known that the determinant can be uniquely characterized as the function that satisfies the following two identities for any pair of $n\times n$ matrices $X,Y$ and any (upper or lower) triangular matrix $ Z $ with $z_{11},\dots, z_{nn}$ on the diagonal:\mardel{This should have a reference.}
\begin{eqnarray} \label{eq:detdef1} \det(X\cdot Y)&=& \det(X)\cd\det(Y), \\
\det(Z)&=& z_{11}\cdots z_{nn}.
\label{eq:detdef2}
\end{eqnarray}
Moreover, other properties of the determinant, such as the cofactor expansion, easily follow from \eqref{eq:detdef1} and \eqref{eq:detdef2}.

The main goal of this paper is to prove the following theorem:

\begin{theorem}[Main theorem]\label{thm: main det}
For any field $\F$:
\begin{enumerate}[(i)]
\item There exists a circuit $\det$ such that \eqref{eq:detdef1} and \eqref{eq:detdef2} have polynomial-size $\PC(\F)$ proofs. Moreover, every\footnote{We assume that the product $z_{11}\cdots z_{nn}$ in (\ref{eq:detdef2}) is written as a formula of depth $O(\log n)$.} circuit in the proof has depth at most $O(\log^2(n))$.
\item There exists a formula $\det$ such that \eqref{eq:detdef1} and \eqref{eq:detdef2} have  $\PF(\F)$ proofs of size $n^{O(\log n)}$.
\end{enumerate}
\end{theorem}

As mentioned before, a large part of the construction is not related directly to the determinant. It is rather a series of structural theorems about the systems $\PF$ and $\PC$. These are obtained by reproducing classical results in arithmetic circuit complexity in the setting of arithmetic proofs (for a recent survey on arithmetic circuit complexity see \cite{SY10}).  The most important of those results is showing that $\PC$ proofs can be balanced, in the sense that $ \PC$ proofs of size $ s $ (of polynomially bounded syntactic degree equations) can be polynomially simulated by $\PC$ proofs in which each circuit has depth $ O(\log^2 s)$.

We do not know whether it is possible to prove Theorem \ref{thm: main det} directly, perhaps by formalizing the elegant algorithm of Berkowitz  \cite{Ber84}. One advantage of the algorithm is that, being division-free,  it would dispense of Theorem \ref{thm: main divisions} and allow to generalize Theorem \ref{thm: main det} to an arbitrary commutative ring (as opposed to a field). We also admit that working with circuits and proofs with divisions turned out to be quite tedious. However, our construction is intended to emphasize general properties of arithmetic proof systems, and the structural theorems are in fact our main contribution.

\subsection{Balancing $\PC $ proofs and simulating $\PC$ by $\PF$}
In the seminal paper \cite{VSB+83}, Valiant et al.~showed that if a polynomial $f$ of degree $d$ can be computed by an arithmetic circuit of size $s$, then $f$ can be computed by a circuit of size $\poly(s,d)$ and depth $O(\log s\log d +\log^2 d)$. This is a strengthening of an earlier result by Hyafil \cite{Hya79}, showing that $f$ can be computed by a formula of size $(s(d+1))^{O(\log d)}$. We will show that those results can be efficiently simulated within the framework of arithmetic proofs.

Instead of the degree of a polynomial, we focus on the syntactic degree of a circuit. Let $F$ be an arithmetic circuit. \emph{The syntactic degree} of $F$, $\deg F$, is defined as follows:
\begin{enumerate}[(i)]
\item If $F$ is a field element or a variable, then $\deg F=0$ and  $ \deg F=1$, respectively;

\item $\deg(F\cplus G)= \max(\deg F, \deg G)$, and  $\deg(F\ctimes G)= \deg F+\deg G$.
\end{enumerate}
The \emph{syntactic degree of an equation $F=G$} is $\max(\deg F, \deg G)$, and the \emph{syntactic degree of a proof $S$} is the maximum of the syntactic degrees of equations in $S$. If $ F $ is a circuit and $ u $ is a node in $ F $ we also write $ \deg(v) $ to denote $ \deg F_v $.

In accordance with \cite{VSB+83}, we will construct a map $[\cdot]$ that maps any given circuit $F$ of size $s$ and syntactic degree $d$ to a circuit  $[F]$ computing the same polynomial, such that $\tr{F}$ has  size $\poly(s,d)$ and depth $O(\log s\log d +\log^2 d)$.
We will show the following:

\begin{theorem}\label{thm: main balancing} Let $F,G$ be circuits of syntactic degree at most $ d$ such that $F=G$ has a $\PC$ proof of size $s$. Then:
\begin{enumerate}[(i)]
\item  The equation $\tr{F}=\tr{G}$ has a $\PC$ proof of size $\poly(s,d)$ and depth $O(\log s\cd\log d +\log^2 d)$.
\item If $F,G$ have depth at most $k$ then $F=G$ has a $ \PC $ proof of size $\poly(s,d)$ and depth $O(k+\log s\cd\log d +\log^2 d) $.
\end{enumerate}
\end{theorem}

We also obtain the following simulation of $\PC$ by $\PF$:

\begin{theorem}\label{thm: main simulation} Assume that $F,G$ are formulas of syntactic degree $\leq d$ such that $F=G$ has a $\PC$ proof of size $s$.
Then  $F=G$ has a $\PF$ proof of size $(s(d+1))^{O(\log d)}\leq s^{O(\log s)}$.
\end{theorem}

This simulation is polynomial if $F$ and $G$ have a constant syntactic degree.
Let us emphasize that the syntactic degree of a formula of size $s$ is at most $s$, and hence the simulation is at most \emph{quasipolynomial}.

\para{Homogenization and degree bound in arithmetic proofs.}
One ingredient of Theorems \ref{thm: main balancing} and \ref{thm: main simulation} is to show that using circuits of high syntactic degree cannot significantly shorten $\PC$ proofs. That is, if we want to prove an equation of syntactic degree $d$, we can without loss of generality  use only circuits of syntactic degree at most $d$. This result is the proof-theoretic analog of a result by Strassen, who showed how to separate arithmetic circuits into their homogeneous parts (implicit in \cite{Str73}).

We say that a circuit $F$ is \emph{syntactically homogeneous}, if for every sum-gate $u_{1}+u_{2}$ in $F$ we have $\deg (u_{1})= \deg (u_{2})$.%
\ignore{or $\deg F_{u_{e}}=0$ and $\widehat F_{u_{e}}=0$  for some $e\in \{1,2\}$. In other words, $F_{u_1}, F_{u_2}$ have the same syntactic degree, or one of them is a {\emph{constant circuit} (i.e., it contains only field constants) computing the zero polynomial.}
An equation $F_{1}=F_{2}$ is \emph{homogeneous}, if $F_{1}$ and $F_{2}$ are syntactically homogeneous, and either $\deg F_{1}=\deg F_{2}$, or $F_{e}=0$ for some $e\in \{1,2\}$. A $\PC$-proof $S$ is \emph{homogeneous}, if it contains only homogeneous equations.} For a circuit $F$ and a number $k$, we introduce a circuit $F^{(k)}$ which computes the syntactically $k$-homogeneous part of $F$ (see Section \ref{sec:bounding} for the definition). The \emph{syntactic degree of a $\PC $ proof} is the maximal syntactic degree of a circuit appearing in it. \new{We show the following:}

\begin{proposition} \label{prop: main hom} 
Assume that $F=G$ has a $\PC$ proof of size $s$. Then
\begin{enumerate}[(i)]
\item $F^{(k)}=G^{(k)}$ has a $\PC$ proof of size $s\cdot \poly(k)$ and a syntactic degree at most $ k$, for any $ k$.
\item If $\deg(F), \deg(G)\leq d$ then $ F=G$ has a $\PC$ proof of syntactic degree at most $d$ and size $s\cdot \poly(d)$.
\end{enumerate}
\end{proposition}

\subsection{Circuits and proofs with division}\label{results: divison}
We denote by $\F(X)$ the field of formal rational functions in the variables $X$ over the field $\F$.
It is convenient to extend the notion of a circuit so that it computes rational functions in $\F(X)$. This is done in the following way:  \emph{a circuit with division} $ F $ is a circuit which may contain an additional type of gate with fan-in $1$, called an \emph{inverse} or a \emph{division} gate, denoted $(\cd)^{-1}$. If a node $v$ computes the rational function $f$, then $v^{-1}$ computes the rational function $1/f$. Moreover, we require that for every division node $v^{-1}$ in $F$, $v$ does not compute the zero rational function. If no division gate computes the zero rational function we say that $ F $ \emph{is defined}, and otherwise, we say that $ F $ is \emph{undefined}. One should note, for instance, that the circuit  $(x^{2}+x)^{-1}$ over $ GF(2) $ is defined, since $x^{2}+x$ is \emph{not} the zero rational function (although it vanishes as a function over $ GF(2)$).

We define the system $\PI(\F)$, operating with equations $F=G$ for $F$ and $G$ circuits with division computing rational functions in $ \F(X)$.  First, we extend the axioms of $\PC(\F)$ to apply to circuits with division. \new{Second, we add}
 the following \iddosml{new axiom to the axioms of $\PC(\F)$}:
\[
        {\rm D}\mez\mez F\cdot F^{-1}= 1 \,,
~~\mbox{provided that $F^{-1}$ is defined.}
\]

\begin{remark}
The system $\PI(\F)$ polynomially simulates the rule \mardel{I think the rule require that \(F^{-1},G^{-1}\) are defined.} 
\[
 \frac{F=G}{F^{-1}=G^{-1}}\,.
\]
Moreover, the identities $(F^{-1})^{-1}=F $ and $  (F\cdot G)^{-1}=G^{-1}\cdot F^{-1}$ have linear size proofs in $\PI(\F)$.
\end{remark}
As before, we sometimes suppress the explicit dependence on the field in $ \PI(\F) $ \iddosml{whenever} the relevant statement is field independent.
\QuadSpace

Strassen \cite{Str73} showed that division gates can be eliminated from arithmetic circuits computing polynomials over large enough fields, with only a polynomial increase in size. We will show the proof-theoretic analog of Strassen's result over \emph{arbitrary} fields, namely that $\PC(\F)$ polynomially simulates $\PI(\F)$ for any field $ \F$, \iddosml{in the following sense}:

\begin{theorem}\label{thm: main divisions}
Let $ \F $ be any field and assume that $F $ and $ G$ are circuits without division gates such that $\deg F, \deg G\leq d$. Suppose that $F=G$ has a $\PI(\F)$ proof of size $s$. Then $F=G$ has a $\PC(\F)$ proof of size $s\cdot\poly(d)$.
\setcounter{thm1}{\thetheorem}
\end{theorem}

A corollary of Theorem \ref{thm: main divisions} is that $\PC(\F)$ polynomially simulates the rule
\[\frac{F\cd G=0}{F=0}\,, \mez \hbox{ if }\,\,\, \wh{G}\not=0\,\]
provided the syntactic degree of $G$ is polynomially bounded.

To prove Theorem \ref{thm: main divisions}, we first assume that the underlying field $\F$ has an exponential size. Under this assumption, we cannot eliminate division gates in $GF(2)$ which is, from the Boolean proof complexity viewpoint, the most interesting field. To deal with small fields and specifically $ GF(2) $ we have to show how to simulate large fields in small ones, as we explain in what follows.

\para{Simulating large fields in small fields.}
The idea behind simulating large fields in small ones is to treat the elements of \new{$GF(p^{n})$ as $n\times n$ matrices over $GF(p)$}. This enables one to simulate computations and proofs over $GF(p^{n})$ by those over $GF(p)$. We  prove the following:

\begin{theorem}\label{main: large field}
Let $p$ be a prime power and $n$ a natural number and let $F,G$ be circuits over $GF(p)$. Assume that $F=G$ has a $\PC(GF(p^{n}))$ proof of size $s$. Then $F=G$ has a $\PC(GF(p))$ proof of size $s\cdot \poly(n)$.
\end{theorem}

\subsection{The determinant as a rational function and as a polynomial}
To prove the main theorem (Theorem \ref{thm: main det}) one needs to construct a circuit (and a formula) computing the determinant polynomial which can be used efficiently inside arithmetic proofs. We first compute the determinant as a \emph{rational function}, using a circuit with divisions denoted $ \Det(X) $,  and show that $ \PI$ admits short proofs of the properties of $ \Det(X) $.
This is achieved by  defining $\Det(X)$ in terms of the matrix inverse $X^{-1}$ and inferring properties of $\Det$ from the identities $X\cdot X^{-1}=X^{-1}X=I$, which are shown to have polynomial-size $\PI$ proofs. The argument is basically a Gaussian elimination.

However, we cannot yet conclude Theorem \ref{thm: main det} which speaks about (division-free) $ \PC $ proofs (it is worth mentioning that we also cannot yet conclude the short \NCT-Frege proofs for the determinant identities, because $\PI $ proofs do not immediately correspond  to propositional Frege proofs). Theorem \ref{thm: main divisions} cannot be directly applied because it allows to eliminate division gates in $\PI $ proofs only if the equations proved are themselves division-free.
We therefore proceed to construct a division-free circuit $ \det(X) $, computing the determinant as a \emph{polynomial}. Assuming we can prove efficiently in $ \PI $ that $ \det(X)=\Det(X) $, we are done, since we can now eliminate division gates from $ \PI $ proofs of division-free equations, using Theorem \ref{thm: main divisions}. To this end, we define the $ \det(X) $ polynomial as the $n$th term of the Taylor expansion of $\Det (I+zX)$ at $ z=0 $. This enables us to demonstrate short proofs of $ \det(X)=\Det(X) $ and conclude the argument.
\HalfSpace

\subsection{Applications}
Equipped with feasible proofs of the determinant identities, short proofs of several related identities follow. Cofactor expansion of the determinant and a version of Cayley-Hamilton theorem will be given in Section \ref{sec:applications}. Another example is the formula completeness of the determinant.
In \cite{Val79:ComplClass}, Valiant showed that every formula of size $s$ can be written as a projection of a determinant of a matrix of a linear dimension. We can conclude that this holds feasibly already  in $\PC$:

\begin{proposition}\label{prop:Valiant} Let $F$ be a formula of size $s$. Then there exists a matrix $M$ of dimension $2s\times 2s$ whose entries are variables or elements of $\F$ such that the identity
\[F=\det(M)\]
has a polynomial-size $ O(\log^2 s)$-depth $\PC(\F)$ proof (and hence also a quasipolynomial-size $\PF(\F)$ proof), \new{where $\det $ is the circuit (resp. the formula) from  Theorem \ref{thm: main det}}.
\end{proposition}

In this paper we are mainly interested in proofs with \emph{no} assumptions other than the axioms A1-A10.  Nevertheless, we can introduce the notion of a \emph{proof from assumptions} as follows: let $S$ be a set of equations. Then \emph{a $ \PC$ proof from the assumptions $ S $} is a proof that can use equations in $S$ as additional axioms (and similarly for $ \PF $ proofs from assumptions).
Proofs from assumptions are far less well-behaved than \new{standard arithmetic proofs}. For instance, neither Theorem \ref{thm: main simulation} nor Theorem \ref{thm: main divisions} hold for proofs from a general nonempty set $ S $ of assumptions.
We now give an important example of a proof from assumptions.

Given a pair of $ n\times n$ matrices $X,Y$, recall that the expressions $XY= I$ and  $YX=I$, are abbreviations for the list of $n^2$ equalities between the appropriate entries.
\iddosml{(We write $I_n$ to denote the $n\times n$ identity matrix.)}


\begin{proposition}\label{prop:XY}
Let $ \F $ be any field. The equations $YX=I_{n}$ have polynomial-size and $ O(\log^2 n) $-depth $\PC(\F)$ proofs from the equations $XY=I_{n}$.
In the case of $\PF(\F)$, the proof has a quasipolynomial-size.
\end{proposition}

\para{Determinant identities in \NCT-Frege and Frege systems.}
When considering the field $\F$ to be $GF(2)$, there is a close connection between our proof systems and the standard propositional proof systems. Consider the propositional proof systems Frege (\new{$F$}),
 extended Frege ($EF$) and circuit Frege ($CF$). For the definitions of Frege and extended Frege see \cite{Kra95} and for the definition of circuit Frege see \cite{Jer04}, where it is also shown that $CF$ and $EF$ are polynomially equivalent.

For simplicity, we shall assume that $F$, $EF$  and $CF$ are all in the Boolean basis \new{ $+,\cdot, 0,1$  (addition and multiplication modulo $ 2 $, logical equivalence, and the two Boolean constants)\footnote{\new{Note that by Reckhow's result, as stated in \cite{Kra95}, the particular choice of basis is immaterial. We could also have $\equiv$  as a primitive.} }. Then every arithmetic circuit \emph{is automatically also a Boolean circuit}, and an equality like $G=H$ can be interpreted as the logical equivalence $G\equiv H$, written as $(G+H)+1$.
 Hence $\PF(GF(2))$ and $\PC(GF(2))$ can be considered as \emph{fragments} of $F$ and $CF$, respectively: the finite set of (schematic) axioms and rules of $\PF(GF(2))$ now serve as Frege axioms and rules, and similarly for $\PC(GF(2))$. Note that $x^{2}= x$ is a propositional tautology but not a polynomial identity, and hence $F$ and $CF$ are (expressively) stronger than their arithmetic counterparts.}
In fact, one can polynomially simulate the full $F$ or $CF$ systems by adding the following new axiom
\[G^2=G\]
to $\PF(GF(2))$ or  $\PC(GF(2))$, where $ G $ is any formula or  a circuit, respectively. \new{To see this, it is sufficient to show that the augmented systems are complete with respect to propositional tautologies: they prove $F=1$ whenever $F$ evaluates to $1$ on every $0,1$-input. }

This means that upper bounds in $\PF(GF(2))$ and $\PC(GF(2))$ are in fact upper bounds in $F$  and $CF$  (and hence also in $EF$), respectively.

In what follows $ XY=I_{n}$, and similarly $YX=I_{n}$, denote the conjunction of $n^2$ formulas of the form
$(x_{i,1}\cdot y_{1,j}+\dots +x_{i,n}\cdot y_{n,j})\equiv \delta_{ij}$, \new{where $+, \cdot$ are addition and multiplication modulo $2$, respectively, $\equiv$ is the logical equivalence,  and $\delta_{ij}\in \{0,1\}$ is given by $\delta_{ij}=1$ iff $i=j$.}
We have the following:

\begin{theorem}\label{thm:INV in NC2-Frege}\mbox{}
\begin{enumerate}
\item The properties of the determinant  as in Theorem \ref{thm: main det} (interpreted as Boolean tautologies \new{over $GF(2)$}) have polynomial-size circuit Frege proofs,  with every circuit of depth at most $O(\log^{2}n)$. In the case of Frege, the proofs have quasipolynomial-size.
        \label{it:det in CF}
\item The implication $(XY=I_{n})\rightarrow (YX=I_{n})$ has a polynomial-size circuit Frege \new{proof}, with every circuit of depth at most $O(\log^{2}n)$, and a quasipolynomial-size Frege proof.
        \label{it:Inv in CF}
\end{enumerate}
\end{theorem}
\begin{proof}
Part  \ref{it:det in CF} is a direct consequence of Theorem \ref{thm: main det} and \ref{it:Inv in CF} of Proposition \ref{prop:XY}, both using the fact that proofs in $ \PC(GF(2))$ and $\PF(GF(2))$ can be interpreted as proofs in circuit Frege and Frege, respectively.
\end{proof}

A family of polynomial-size CF proofs in which every proof-line $ G $ is of depth $ O(\log^2 |G|) $, is also called an \NCT-Frege proof. Hence, Theorem \ref{thm:INV in NC2-Frege} states that \NCT-Frege has polynomial-size proofs of the propositional tautologies $(XY=I)\rightarrow (YX=I)$.

Theorem \ref{thm:INV in NC2-Frege} thus settles an important open problem in proof complexity and feasible mathematics, namely, whether basic properties of the determinant like $ \det(A)\cd\det(B) = \det(AB) $ and the cofactor expansion (see Proposition \ref{prop:cofactor}), as well as the hard matrix identities, have polynomial-size proofs in a proof system which corresponds to the circuit class \NCT.

\begin{remark}
We believe that Theorem \ref{thm:INV in NC2-Frege} can be extended to any finite field or the field of rationals (after encoding field elements as Boolean strings). For finite fields, this is rather straightforward. In the rational case, one would have to show that the $\PC(\Q)$ proofs constructed in Theorem \ref{thm: main det} involve only constants whose Boolean representation is polynomial.
\end{remark}
\FullSpace

%
%
%

\section{Homogenization and bounding the degree in $\PC(\ring)$ proofs}\label{sec:bounding}
In this section we wish to construct the circuits $F^{(k)}$ computing the $k$-homogeneous part of $F$ and prove Proposition \ref{prop: main hom}.
First, let us say that a circuit $F$ is \emph{non-redundant}, if either $F$ is the constant $0$, or $F$ does not contain the constant $0$ at all.
Any circuit $F$ can be transformed to a non-redundant circuit $F^{\sharp}$ as follows: successively replace all nodes of the form $u+0$, $0+u$ by $u$ and $u\cdot 0$, $0\cdot u$ by $0$, until no such replacement is possible.

Let $k$ be a natural number. We define $F^{(k)}$ as follows. For every node $u$ in $F$, introduce $k+1$ new nodes $u^{(0)},\dots, u^{(k)}$.
\begin{enumerate}
\item Assume $u$ is a leaf. Then, $u^{(0)}:=u$, in case $ u $ is a field element, and $u^{(1)}:=u$ in case $ u$ is a variable, and $u^{(i)}:=0$ otherwise.
\item If $u= u_1+u_2$, let $u^{(i)}:= u_1^{(i)}+u_2^{(i)}$, for every $i=0,\dots, k$.
\item If $u= u_1\cdot u_2$, let $u^{(i)}:= \sum_{j=0}^iu_1^{(j)}\cdot u_2^{(i-j)}$.
\end{enumerate}
Finally, we define $F^{(k)}$ as the circuit $G^{\sharp}$, where $G$ is the circuit with the output node $w^{(k)}$ and $w$ is the output node of $F$.

Note the following:
\begin{enumerate}[(1)]
\item $F^{(k)}$ has size
 $O(s(k+1)^2))$, where $s$ is the size of $F$.
 \item $F^{(k)}$ is a syntactically homogeneous non-redundant circuit. Its syntactic degree is either $k$, or \new{$F$} is the constant $0$.
\end{enumerate}

\begin{notation}\label{rem:big sum}
We allow circuits and formulas to use only sum gates with fan-in two. An expression $\sum_{i=1}^k x_i$ is an abbreviation for a  formula of size $O(k)$ and depth $O(\log k)$ with binary sum gates. For example, define $\sum_{i=1}^k x_i:= \sum_{i=1}^{\lfloor k/2\rfloor} x_i+ \sum_{i=\lceil k/2\rceil}^k x_i\,.$ One can see that basic identities such as
\[\sum_{i=1}^k x_i= \sum_{i=1}^m x_i+ \sum_{i=m+1}^k x_i\,, \,\,\hbox{ or }\,\, y\cdot \sum_{i=1}^k x_i=\sum_{i=1}^k yx_i \]
have $\PF$ proofs of size $O(k^{2})$ and depth $O(\log k)$.
\end{notation}

 \begin{lemma}\label{lem:bounded degree}
 Let $F_1, F_2$ be circuits of size $\leq s$ and $k$ a natural number. The following have proofs of size $s\cdot\poly(k)$ and syntactic degree $\leq k$.
\begin{enumerate}
\ignore{\item $F^{(0)}=F$, if $F\in\ring$.  $F^{(1)}=F$, if $F$ is a variable.  $F^{(k)}=0$, if $F$ is a ring element or a variable and $k>1$.}
\item  $(F_1\cplus F_2)^{(k)}= F_1^{(k)}+F_2^{(k)}$,
\item   $(F_1\ctimes F_2)^{(k)}= \sum_{i=0}^kF_1^{(i)}\cdot F_2^{(k-i)}$.
\end{enumerate}
 \end{lemma}

\begin{proof} It is easy to see that for any circuit $H$ of size $s$, $H=H^{\sharp}$ has a proof of size $O(s)$.
This, and the definition of $F^{(k)}$, gives
$(F_1\cplus F_2)^{(k)}= F_1^{(k)}\cplus F_2^{(k)}$. Hence $(F_1\cplus F_2)^{(k)}=F_1^{(k)}+ F_2^{(k)}$ by axiom C1. Since $F_1^{(k)}, F_2^{(k)}, (F_1\cplus F_2)^{(k)}$ \iddosml{all have circuit} size $O(s(k+1))^2$, we obtain (i). Part (ii) is similar.
\end{proof}

\begin{lemma}\label{lem:bounded degree2}
If $F$ is a circuit with syntactic degree $\leq d$ and size $
s$ then \[F= \sum_{k=0}^d F^{(k)}\]
has a $\PC(\F)$ proof of syntactic degree $\leq d$ and size $s\cdot\poly(d)$.
\end{lemma}

\begin{proof} For every node $u$ in $F$, construct a proof of   $F_u= \sum_{k=0}^{\deg (u)}F_u^{(k)}$. This is done by induction on depth of $u$. If $u$ is a leaf, this stems from the definition of $F_u^{(k)}$, and if $u=u_1+u_2$ or $u=u_1\cdot u_2$, it is an application of the previous lemma.
\end{proof}

\bigskip
\noindent
\emph{Proof of Proposition \ref{prop: main hom}. }
Part (ii) follows from (i) by Lemma \ref{lem:bounded degree2}, hence it is sufficient to prove part (i).
 Let us first show that if $F=G$ is an axiom of $\PC(\ring)$ of size $s$ then $F^{(k)}=G^{(k)}$ has a proof of size $s\cdot \poly(k)$ and syntactic degree $\leq k$.  This is an application of Lemma \ref{lem:bounded degree}. Let $c$ be the constant such that equations (i) and (ii) in  Lemma \ref{lem:bounded degree} have proofs of size $O(s\cdot (k+1)^c)$.

 The lemma gives a proof $(F_1\cplus F_2)^{(k)}=(F_1+ F_2)^{(k)} $ and $(F_1\ctimes F_2)^{(k)}=(F_1\cdot F_2)^{(k)} $, as required for the axioms C1 and C2.

 Axioms A1 and A10 are immediate.
  For the other axioms, consider for example the axiom $F_1\cdot (F_2\cdot F_3)= (F_1\cdot F_2)\cdot F_3$, where the circuits have size $\leq s$. We have to construct a proof of
\begin{equation}(F_1\cdot (F_2\cdot F_3))^{(k)}= ((F_1\cdot F_2)\cdot F_3)^{(k)}\,.\label{eq:0}\end{equation}
By part (ii) of Lemma \ref{lem:bounded degree} the equations
\begin{eqnarray}
(F_1\cdot (F_2\cdot F_3))^{(k)}&=&\sum_{i=0}^k F_1^{(i)}\left( \sum_{j=0}^{k-i} F_2^{j} F_3^{k-i-j}\right) \label{eq:1}\\
((F_1\cdot F_2)\cdot F_3)^{(k)}& =& \sum_{i=0}^k \left( \sum_{j=0}^{i} F_1^{j} F_2^{i-j}\right)\cdot F_3^{(k-i)} \label{eq:2}\,,
\end{eqnarray}
can be proved by proofs with size roughly $s\cdot (k+1)^c\cdot (k+1)$.
In $\PC(\ring)$, the right hand sides of both (\ref{eq:1}) and (\ref{eq:2}) can be written as
$\sum_{ i+j+l=k} F_1^{(i)}F_2^{(j)}F_3^{(l)}$, by a proof of size roughly $s(k+1)^4$ . This gives the proof of (\ref{eq:0}) of size $s\cdot \poly(k)$.

Next, assume that $F=G$ is derived from the equations $F_1=G_1, F_2=G_2$ by means of the rules R1-R4, and we need to construct the proof of   $F^{(k)}=G^{(k)}$  from the set of equations $F_1^{(i)}=G_1^{(i)}, F_2^{(i)}=G_2^{(i)}, i=0,\dots k$. The hardest case is the rule
\[\frac{F_1=G_1\mez F_2=G_2 }{F_1\cdot F_2= G_1\cdot G_2}\,.\]
We have to prove  $(F_1\cdot F_2)^{(k)}= (G_1\cdot G_2)^{(k)}$. By Lemma \ref{lem:bounded degree}, we have proofs of $(F_1\cdot F_2)^{(k)}=\sum_{i=0,\dots k} F_1^{(i)}F_2^{(k-i)}$ and
$(G_1\cdot G_2)^{(k)}=\sum_{i=0,\dots k} G_1^{(i)}G_2^{(k-i)}$.
Hence $(F_1\cdot F_2)^{(k)}= (G_1\cdot G_2)^{(k)}$ can be proved from the assumptions $F_1^{(i)}=G_1^{(i)}, F_2^{(i)}=G_2^{(i)}, i=0,\dots k$. The proof has size roughly $s\cdot(k+1)^c (k+1)$.
\qed


\newcommand{\N} {\mathbb{N}}

\section{Balancing $ \PC$  proofs}\label{sec:balancing proofs}
In this section we prove Theorem \ref{thm: main balancing} which is a proof-theoretic analog of the following result:

\begin{theorem}[Valiant et al. \cite{VSB+83}]\label{thm:Val et al}
Let $ F$ be an arithmetic circuit of size $ s $ computing a polynomial $ f $ of degree $ d $. Then there exists an arithmetic circuit $[F]$ computing $ f $  with depth $ O(\log^2 d + \log s \cd\log d) $ and size  $ \poly(d,s)$.
\end{theorem}

We first give an outline of the construction of $\tr{F}$, which closely follows that in \cite{VSB+83} (we also refer the reader to \cite{RY08-balancing} for an especially clear exposition). We emphasize that in our case, the relevant parameter is the \emph{syntactic} degree of $F$: $\tr{F}$ will have size $\poly(s,d)$ and depth  $O(\log^2 d + \log s \cd\log d) $, where $d$ is the syntactic degree of $F$.

We write $ u\in F $ to mean that $ u $ is a node in the circuit $F $.
The following definition is important for the construction of balanced circuits:
let $ w,v $ be two nodes in $ F $.  We define the \emph{polynomial}  $ \partial w F_v  $ as follows:
\[
\partial w F_v :=
          \left\{
                \begin{array}{ll}
                  0, & \hbox{if $ w\not\in F_v$,} \\
                  1, & \hbox{if $ w=v $ ,  and otherwise:} \\
                  \partial w F_{v_1} + \partial w F_{v_2}  , & \hbox{$ v=v_1+v_2$;} \\
                  (\partial w F_{v_1}) \cd F_{v_2} , & \hbox{if either $ v=v_1\cd v_2 $ and $ \deg(v_1)\geq\deg(v_2) $,}\\
                    & \hbox{or $ v=v_2\cd v_1 $ and $ \deg(v_1)> \deg(v_2) $.}
                \end{array}
          \right.
\]
The idea behind this definition is the following: let $ w,v $ be two nodes in $ F $ such that  $2 \deg(w) > \deg (v) $.  Then
for any product node $v_{1}\cdot v_{2}\in F_{v}$, $w$ can be a node in at most one of $F_{v_{1}}, F_{v_{2}}$, namely the one of a higher syntactic degree. If we replace the node $w$ in $F_{v}$ by a new variable $z$, $F_{v}$ then computes a polynomial $g(z,x_{1},\dots,x_{n})$ which is linear in $z$, and $\partial w F_{v}$ is the usual partial derivative $\partial z g$.

It is not hard to show the following:

\begin{claim}\label{prop:deg fv - deg w}
Let $ w,v $ be two nodes in a circuit $ F $. Then the polynomial $ \partial w F_{v}$ has degree at most $ \deg(v)-\deg(w) $.
\end{claim}

In order to construct $\tr{F}$, we can assume without loss of generality that $ F $ itself is a syntactic homogenous circuit of size $ s'=O(d^2\cd s) $.
This is because a circuit of size $s$ and syntactic degree $d$ can be written as a sum of $d+1$ syntactically homogeneous circuits of size at most $s'$ and syntactic degree at most $d$.
Now the construction proceeds by induction on $ i=0,\ldots,\lceil \log d \rceil $. In each step $ i=0,\ldots,\lceil \log d \rceil $ we construct:
\begin{enumerate}
\item  Circuits computing $ \wh{F}_v $, for all nodes $ v $ in $ F $ with $ 2^{i-1}<\deg(v)\le 2^i $;
\item  Circuits computing $ \partial w F_{v}  $, for all nodes $ w,v $ in $ F $ with $ 2^{i-1}<\deg(v)-\deg(w)\le 2^i $ and $ \deg(v)<2\deg(w) $.
\end{enumerate}
Each step adds depth  $ O(\log s') $, which at the end amounts to a depth $ O(\log^{2}d+\log d \cd \log s ) $ circuit. Furthermore, each node $ v $ in $ F $ adds $ O(s') $ nodes in the new circuit and each pair of nodes $ v,w $ in $ F $ adds also $ O(s') $ nodes in the new circuit. This  finally amounts to a circuit of size $ O(s'^3)=O(d^6\cd s^3) $.

\bigskip
Let us now give the formal definition of $\tr{F}$. First, for a circuit $ G $ and a natural number $m$, let
\[
        {\cal B}_m(G) :=
                        \left\{
                                        t \in G \;:\; t = t_1\cd t_2 ,  \deg(t)>m \mbox{ and } \deg(t_1),\deg(t_2)\le m
                        \right\}.
\]

\para{Definition of  $[F] $.}
Let $ F $ be an arithmetic circuit of syntactic degree $ d $.

If $F$
 is not syntactic homogenous, let
\[
        [F]:=[F^{(0)}]+\ldots+[F^{(d)}]\,.
\]

Otherwise, assume that $ F $ is a syntactically homogenous circuit of degree $ d$. For any \emph{node} $ v\in F $ we introduce the corresponding \emph{node} $ [F_v] $ in $ [F] $ (intended to compute the polynomial $ \wh{F}_v $); and for any pair of nodes $ v,w \in F $ such that $ 2\deg(w) >\deg(v) $, we introduce the node $ [\partial w F_{v} ]$ in $ [F] $ (intended to  compute the polynomial $ \partial w F_{v} $).

The construction is defined by induction on $ i =0,\ldots,\lceil \log d \rceil $, as follows:

\para{Part (I):} Let $ v\in F $: \mbox{}
\QuadSpace

\case 1
Assume that \new{$ \deg(v)\leq1 $},
 then $ F_v $ computes a linear polynomial $ a_1x_1+\ldots+a_nx_n+b $ (where, by homogeneity of $ F $, $ b\neq 0$ \new{implies that}
  all $ a_i $'s equal $ 0 $). Define
\[
        [F_v]:= a_1x_1+\ldots+a_nx_n+b .
\]
\QuadSpace

\case 2 Assume that for some $ 0 \le i \le \lceil \log(d) \rceil $:
\[
        2^i < \deg(v) \le 2^{i+1}.
\]
Put $ m=2^i $, and define
\[
        [F_v] := \sum_{t\in {\cal B}_m(F_v)} [\partial t F_{v}]\cd [F_{t_1}]\cd[F_{t_2}],
\]
\new{where $t_{1}, t_{2}$ are nodes such that $ t=t_1\cd t_2 $.}
 (Note that here $ [\partial w F_v], [F_{t_1}] $ and $[F_{t_2}]$ are \emph{nodes}.)

\para{Part (II):} Let $ w,v $ be a pair of nodes in $ F $ with $ 2\deg(w) >\deg(v) $:\mbox{}
\QuadSpace

\case 1 Assume $ w $ is not a node in $ F_v $. Define
\[
        [\partial w F_v]:=0.
\]

\case 2 Assume that $ w $ is in $ F_v $ and $ 0\le \deg(v)-\deg(w)\le 1 $. Thus, by Claim \ref{prop:deg fv - deg w}, the polynomial $ \partial w f_{v} $ is a linear polynomial $ a_1x_1+\ldots+a_nx_n+b $. Define
\[
        [\partial w F_v]:= a_1x_1+\ldots+a_nx_n+b .
\]
\QuadSpace

\case 3 Assume that $ w $ is in $ F_v $ and that for some $ 0 \le i \le \lceil \log(d) \rceil $:
\[
2^i < \deg(v)-\deg(w)\le 2^{i+1} .
\]
\QuadSpace

Put $  m = 2^i + \deg(w)$. Define:
\[
        [\partial w F_{v}] := \sum_{t\in{\cal B}_m(F_v) } [\partial t F_{v}] \cd[\partial w F_{t_1}]\cd [F_{t_2}]\,,
\]
\new{where $t_{1}, t_{2}$ are nodes such that $ t=t_1\cd t_2 $ and $\deg(t_{1})\geq\deg(t_{2})$, or $t=t_{2}\cdot t_{1}$ and $\deg(t_{2})>\deg(t_{1})$.}
Finally, define $\tr{F}$ as the circuit with the output node $\tr{F_{u}}$, where $u$ is the output node of $F$.

\bigskip
One should make sure that the definition of $\tr{F}$ is well defined, and that it has the correct depth and size:

\begin{lemma}\label{lem: correctness} Let $F$ be a circuit of size $s$ and syntactic degree $d$. Then $\tr{F}$ is a circuit computing $\wh{F}$,  $\tr{F}$ is  of size $\poly(s,d)$ and depth $O(\log^{2}d+\log s\log d)$. {Moreover, every node $[\partial{w} F_{v}]$ in $[F]$ computes the polynomial $\partial w F_{v}$. }
\end{lemma}

\begin{proof} The proof is as in \cite{VSB+83} (see also \cite{RY08-balancing}). We shall give a partial sketch of the proof here, for the benefit of the reader.

First, assume that $F$ is syntactic homogeneous of degree $d$. We need to verify that $\tr{F}$ is well-defined. That is, at stage $i=0,\dots,\lceil\log d\rceil$, we compute all $\tr{F_{v}}$ and $\tr{\partial_{w}F_{u}}$ for all nodes $v,u,w\in F$ such that $2^i < \deg(v) \le 2^{i+1}$ and $2^i < \deg(v)-\deg(u)\le 2^{i+1}$, and
we want to show that the computation uses only nodes computed in previous stages.

Take, for example, Case 2 in Part (I). For any $ t\in {\cal B}_m(F_v)$, $  m<\deg(t)\le\deg(v)\le 2m $. This implies that $ \deg(v)-\deg(t) \le m=2^i $ and $ \deg(t)<2\deg(v) $. Hence, we have already computed $ [\partial t F_v] $. We have also already constructed $ [F_{t_1}], [F_{t_2}] $, since $ \deg(t_1), \deg(t_2)<m=2^i $.

Inspecting the construction, $\tr{F}$ has size $\poly(s)$ and depth $O(\log s\cd \log d)$, given that $F$ is syntactically homogeneous of size $s$ and degree $d$. If $F$ is not syntactically homogeneous, the definition $\tr{F}= \tr{F^{(0)}}+\dots \tr{F^{(d)}}$ gives a circuit of size $\poly(s,d)$ and depth $O(\log^{2} d+\log s\cd \log d)$, since every $F^{(k)}$ has size $O(s\cdot k^{2})$.
 \end{proof}

We need to show that properties of  $\tr{F}$ can be proved inside the system $\PC$.
The key ingredient is given by the following lemma.

\begin{lemma}[Main simulation lemma]\label{lem:main simulation lemma}\label{lem:simulation}
Let $ F_{1}, F_2 $ be circuits of syntactic degree at most $ d $ and size at most $ s $. Then there exist $\PC$ proofs of:
\begin{align}
 [F_1\oplus F_2]  &= [F_1]+[F_2]\, , \ \ \ \     \label{eq:plus}\\
[F_1\otimes  F_2] & = [F_1]\cd[F_2] \,,    \label{eq:prod}
\end{align}
such that the proofs have size $\poly(s,d)$ and depth $O(\log^{2}d+\log d\cd \log s)$.
\end{lemma}

The proof of Lemma \ref{lem:main simulation lemma} is deferred to the end of this section. We now use Lemma \ref{lem:simulation}
 to prove Theorems \ref{thm: main balancing} and \ref{thm: main simulation}.

\begin{theorem}[Theorem \ref{thm: main balancing} restated] Let $F,G$ be circuits of syntactic degrees at most  $d$ such that $F=G$ has a $\PC$ proof of size $s$. Then
\begin{enumerate}
\item \label{bal-jedna} $\tr{F}=\tr{G}$ has a $\PC$ proof of size $\poly(s,d)$ and depth $ O(\log s\cd\log d +\log^2 d)$.
\item \label{bal-dva} If $F,G$ have depth at most $ t$ then $F=G$ has a $\PC$ proof of size $\poly(s,d)$ and depth at most $ O(t+\log s\cd\log d +\log^2 d) $.
\end{enumerate}
\end{theorem}

\begin{proof}
Part \ref{bal-jedna}. Assume that $F=G$ has syntactic degree $d$ and a $\PC$ proof of size $s$. By Proposition \ref{prop: main hom},
$F=G$ has a $\PC$ proof of syntactic degree $d$ and size $s'=s\cdot \poly(d)$.  So let us consider such a proof $S$. By induction on the number of lines in $S$, construct a $\PC$ proof of $\tr{F_1}=\tr{F_2}$, where $F_1=F_2$ is a line in $S$.

Let $m_0$ and $k_0$ be  such that (\ref{eq:plus}) and (\ref{eq:prod}) have $\PC$ proofs of size at most $m_0$ and depth $k_0$, whenever $F_1\cplus F_2$, respectively, $F_1\ctimes F_2$ have size at most $s'$ and syntactic degree at most $d$. By Lemma \ref{lem:main simulation lemma}, we can choose $m_0=\poly(s',d)$ and $k_0=O(\log s'\cd \log d +\log^2 d)$.
\QuadSpace

First, show that if a line $F=H$ in $S$ is a $\PC$ \emph{axiom} then $\tr{F}=\tr{H}$ has a $\PC$ proof of size $c_1m_0$ and depth $c_2 k_0$, where $c_1, c_2$ are some constants independent of $s',d$. The axiom A1 is immediate and the axiom A10 follows from the fact that $[F]= \widehat F$, if $\deg(F)=0$. The rest of the axiom are an application of Lemma \ref{lem:simulation}, as follows. Axioms C1 and C2 are already  the statement of Lemma \ref{lem:simulation}. For the other axioms, take, for example,
\[F_1\cdot (G_1+G_2)= F_{1}\cdot G_1+F_{1}\cdot G_2\, .\]
We are supposed to give a proof of
\[\tr{F_1\cdot (G_1+G_2)}= \tr{F_1\cdot G_1+ F\cdot G_2}\, ,\]
with a small size and depth.
By Lemma \ref{lem:simulation} we have a $\PC$ proof
\[\tr{F_1\cdot (G_1+G_2)}= \tr{F_1}\cdot\tr{ G_1+ G_2} = \tr{F_1}\cdot \tr{G_1}+\tr{F_1}\cdot \tr{G_2}= \tr{F_1}\cdot (\tr{G_1}+\tr{G_2})\,.\]
Lemma \ref{lem:simulation} gives again
\[\tr{F_1}\cdot (\tr{G_1}+\tr{G_2})= \tr{F_1}\cdot \tr{G_1+G_2}= \tr{F_1\cdot (G_1+G_2)}\,.\]
Here we applied Lemma \ref{lem:simulation} to circuits of size at most $s'$, and the proof of $\tr{F_1\cdot (G_1+G_2)}= \tr{F_1\cdot G_1+ F\cdot G_2}$ has size at most, say, $100m_0$ and  depth at most $10k_0$.

An application of rules R1, R2 translates to an application of R1, R2. For the rules R3 and R4,  it is sufficient to show the following: if $S$ uses the rule \[\frac{F_1=F_2\qquad G_1=G_2}{F_1\circ G_1=F_2\circ G_2},\, \circ\in\{\cdot, +\}, \] then there is a proof of $\tr{F_1\circ G_1=F_2\circ G_2}$, of size $c_1m_0$ and depth $c_2k_0$, from the equations $\tr{F_1}=\tr{G_1}$ and $\tr{F_2}=\tr{G_2}$. \new{This is again}
 an application of Lemma \ref{lem:simulation}.

Altogether, we obtain a proof of $\tr{F}=\tr{G}$ of size at most $c_1s'm_0$ and depth $c_2k_0$.
\FullSpace

\ind Part \ref{bal-dva}. Using \ref{bal-jedna},  it is sufficient to prove the following:

\begin{claim*} If $F$ is a circuit with depth $t$, syntactic degree $d$ and size $s$, then $F=[F]$ has a $\PC$ proof of size $\poly(s,d)$ and depth at most $ O(t+\log s\cd\log d +\log^2 d) $.
\end{claim*}
Using Lemma \ref{lem:simulation}, this claim can be easily proved by induction on $s$.
\end{proof}

\begin{theorem}[Theorem \ref{thm: main simulation} restated]\label{thm:5-restated-is now flying} Assume that $F,G$ are formulas of syntactic degree at most $d$ such  that $F=G$ has a $\PC$ proof of size $s$.
Then  $F=G$ has a $\PF$ proof of size $(s(d+1))^{O(\log d)}$.
\end{theorem}

\begin{proof} Recall the definition of the formula $F^{\bullet}$ from \new{Remark} \ref{def:similar}.
It is not hard to show the following:

\begin{claim1}   If $H_1=H_2$ has a $\PC$ proof with $p$ proof lines and depth $k$, then $H_1^\bullet=H_2^\bullet$ has a $\PF$ proof of size $O(p2^k)$.
\end{claim1}

\bigskip
Let $F$ and $G$ be as in the assumption. The previous theorem and Claim 1 give a $\PF$ proof of
\[\tr{F}^\bullet = \tr{G}^\bullet\]
of size $s \cdot 2^{O(\log s \cd\log d+\log^{2}d)}=(s(d+1))^{O(\log d)}$.

To complete the proof, it is sufficient to show that:
\bigskip

\begin{claim2}  If $H$ is a formula of size $s$ and syntactic degree $d$, then $[H]^\bullet=H$ has a $\PF$ proof of size $(s(d+1))^{O(\log d)}$.
\end{claim2}
\bigskip
This is proved by induction on $s$ using Lemma \ref{lem:simulation}.
\end{proof}

\subsubsection*{Proof of Lemma \ref{lem:simulation}}\label{sec: lemma}
It is sufficient to prove the statement, under the assumption that $F_{1}\cplus F_{2}$ and $F_{1}\ctimes F_{2}$ are syntactically homogeneous. This is because of the following: assume that the lemma holds for syntactically homogeneous circuits.
First, note that for any circuit of syntactic degree $d$,
\[
        \tr{F}=\tr{F^{(0)}}+\tr{F^{(1)}}+\dots +\tr{F^{(d)}}
\]
has a proof of size $\poly(s,d)$ and depth $O(\log d\cd\log s+\log^{2}d)$: if $F$ is not syntactically homogeneous, then this stems from the definition of $[F]$; otherwise, $F$ is syntactically homogeneous, and so $[F^{(k)}]$ is the circuit $0$ whenever $k<d$ and it is sufficient to construct the proof of $[F]=[F^{(d)}]$, which can be done by induction on the size of $F$.
Second, if for example $F_{1}\cplus F_{2}$ is not syntactically homogenous, then by definition of $\tr{\cdot}$, we have
\[
        \tr{F_{1}\cplus F_{2}}= \sum_{k=0}^d \tr{(F_{1}\cplus F_{2})^{(k)}}\,,
\]
where $d=\deg(F_{1}\cplus F_{2})$. By the definition of $F^{(k)}$, $(F_{1}\cplus F_{2})^{(k)}$ is a syntactically homogeneous circuit which is either of the form $F_{1}^{(k)}\cplus F_{2}^{(k)}$,
or it is of the form $F_{e}^{(k)}$, if $F_{e'}^{(k)}=0$, $\{e,e'\}=\{1,2\}.$
In both cases we obtain a proof of $[(F_{1}+F_{2})^{(k)}]= [F_{1}^{(k)}]+[F_{2}^{(k)}]$, of small size and depth.
This gives a $\PC$ proof of
 \[\sum_{k=0}^d \tr{(F_{1}\cplus F_{2})^{(k)}}= \sum_{k=0}^d\tr{(F_{1})^{(k)}}+ \tr{(F_{2})^{(k)}}
        = \sum_{k=0}^d \tr{(F_{1})^{(k)}}+ \sum_{k=0}^d \tr{(F_{2})^{(k)}}\,.\]

We thus consider the syntactically homogeneous case. Let $m(s,d)$ and $r(s,d)$ be functions such that
for any circuit $ F $ of syntactic degree $ d $ and size $s$, $\tr{F}$ has depth at most $ r(s,d) $ and size at most $ m(s,d) $. By Lemma \ref{lem: correctness}, we can choose
\[ m(s,d)=\poly(s,d) \mbox{\ \ and \ \ } r(s,d)=O(\log^{2}d+\log d \cd\log s).\]

\begin{notation}
In the following, $[F_v]$ and $[\partial w F_v]$ will denote \emph{circuits}: $[F_v]$ and $[\partial w F_v]$ are the subcircuits of $[F]$ with output nodes $[F_v]$ and $[\partial w F_v]$, respectively; the defining relations between the nodes of $[F]$ (see the definition of $[F]$  above) translate to equalities between the corresponding circuits. For example, if $v$ and $m$ are as in Case 2, part (I) of the definition of $ \tr{F} $, then, using just the axioms C1 and C2, we can prove
\begin{equation}\label{eq:hosafa basof}
[F_v] = \sum_{t\in {\cal B}_m(F_v)} [\partial t F_{v}]\cd [F_{t_1}]\cd[F_{t_2}]\,.
\end{equation}
\new{Here, the left hand side is understood as the circuit $ \tr{F_v} $ in which $ \tr{\partial t F_v}, \tr{F_{t_1}}, \tr{F_{t_2}} $ appear as \emph{subcircuits}, and so can share common nodes, while on the right hand side the circuits have \emph{disjoint nodes}.}
Also, note that if $F$ has size $s$ and degree $d$, the proof of  (\ref{eq:hosafa basof}) has size  $O(s^2m(s,d))$ and has depth $O(r(s,d))$.
\end{notation}
We shall use these kind of identities in the current proof.

 The following statement suffices to conclude the lemma. The recurrence (\ref{eq:recurrence}) below  implies $\lambda(s,d)=\poly(s,d)$ and it is enough to take $ F $ in the statement as either $ F_1\oplus F_2 $ or $ F_1\otimes F_2 $, and $ v $ as the root of $ {F} $.

\para{Statement:}
Let $ F $ be a syntactically homogenous circuit of syntactic degree $d $ and size  $ s $, and let $ i=0,\ldots,\lceil \log d \rceil $. There exists a function $ \lambda(s,i) $ not depending on $F$ with
\begin{equation}\label{eq:recurrence}
   \lambda(s,0)=O(s^4) \quad \hbox{and} \quad \lambda(s,i)\le O(s^4\cd m(s,d))+ \lambda(s,i-1),
\end{equation}
and a $\PC$ proof-sequence $ \Psi_i $ of size at most $ \lambda(s,i) $ and depth at most $ O(r(s,d)) $, such that the following hold:
\QuadSpace

\ind\textbf{Part (I):}
For every node $ v\in{F} $ with
\begin{equation}\label{eq:od ehad}
 \deg(v)\le 2^i,
\end{equation}
$ \Psi_i $ contains the following equations:
\begin{align}
        [F_v] &= [F_{v_1}]+[F_{v_2}]\ , \mbox{\ \ \ \ \ in case $ v = v_1+v_2 $,\ \ \ \ \  and} \label{eq:Fv+}\\
        [F_v] &= [F_{v_1}]\cd [F_{v_2}]\ , \mbox{\ \ \ \ \ in case $ v = v_1\cd v_2 $}.\label{eq:Fv*}
\end{align}

\ind\textbf{Part (II):} For every pair of nodes $ w\neq v\in{F} $, where $ w\in F_v $, and with
\begin{align}
   &
\deg(v)-deg(w)\le 2^i \; \mbox{\ \ \ and}  \label{eq:nizkarti1}\\
        & 2\deg(w)>\deg(v),\label{eq:nizkarti2}
\end{align}
$ \Psi_i $ contains the following equations:
\begin{align}
        [\partial w F_{v}]  &= [\partial w F_{v_1}] + [\partial w F_{v_2}], &\mbox{\ \ \ in case $ v = v_1+v_2 $;}\label{eq:wFv+}\\
        [\partial w F_{v}]  &= [\partial w F_{v_1}] \cd [F_{v_2}] ,& \mbox{\ \ in case $ v = v_1\cd v_2 $ and $\deg(v_1)\ge\deg(v_2)$ }\nonumber\\
           && \mbox{\ \ or $ v = v_2\cd v_1 $ and $\deg(v_1)> \deg(v_2)$.  }
        \label{eq:wFv*}
\end{align}
\HalfSpace

We proceed to construct the sequence $\Psi_{i}$ by induction on $ i $.

\Base $ i=0 $.
We need to devise the proof sequence $ \Psi_0 $.

\para{Part (I).}  Let $ \deg(v)\le 2^0 $. By definition, $\tr{F_{v}}= \sum_{i=1}^n a_i x_i +b$, where $a_{i}$'s and $b$ are field elements.  If $v=v_{1}+v_{2}$, we have also   $\tr{F_{v_{e}}}= \sum_{i=1}^n a_i^{(e)} x_i +b^{(e)}$, for $e=1,2$. Hence the equation $  [F_v] = [F_{v_1}] + [F_{v_2}] $ is the (true) identity:
\[
 \sum_{i=1}^n a_i x_i +b= \sum_{i=1}^n a_i^{(1)} x_i +b^{(1)}
        +\sum_{i=1}^n a_i^{(2)} x_i +b^{(2)}\,,
\]
which has a proof of size $O(s^{2})$ and depth $O(\log s)$ (we assume without loss of generality that $n\leq s$).

In case $ v = v_1\cd v_2 $, either $\deg(v_{1})=0$ or $\deg(v_{2})=0$ and the proof of $ [F_v] = [F_{v_1}]\cd[F_{v_2}] $ is similar.

\para{Part (II).}
Since $ \deg(v)-\deg(w) \le 1 $, we have  $ [\partial w F_{v}] = \sum_{i=1}^n a_i x_i +b $, for some field elements $ a_i $'s and $b $.

In case $ v = v_1+v_2 $, we have $\deg(v_{e})-\deg(w)\leq 1 $ and so $ [\partial w F_{v_e}]=\sum_{i=1}^n a_i^{(e)} x_i + b^{(e)}$, where $ e=1,2 $. The assumption $w\not=v$ and Lemma \ref{lem: correctness}, guarantee that $[\partial w F_{v}]=[\partial w F_{v_{1}}]+ [\partial w F_{v_{2}}]$ is a correct identity, and we can thus proceed as the base case of Part (I) above.

In case   $ v=v_1\cd v_2 $, assume without loss of generality that $\deg (v_{1})\geq \deg (v_{2})$. Again, we have $ [\partial w F_{v_1}]=\sum_{i=1}^n a_i^{(1)} x_i + b^{(1)}$.   From the assumptions, we have that $w\in F_{v_{1}}$, which implies $\deg(v_{1})\geq \deg(w)$ and so $\deg(v_{2})\leq 1$.
Hence $[F_{v_{2}}]= \sum_{i=1}^n a_i^{(2)} x_i + b^{(2)} $. (One can note that
 at least one of $[\partial w F_{v_1}]$ or $[F_{v_2}]$ is constant). Thus we can prove the (correct, by virtue of the assumption $w\not=v$) identity $ [\partial w F_v] =[\partial w F_{v_1}] \cd [F_{v_2}]$ with a $ \PC(\F)$ proof of size $ O(s^2) $  and depth $ O(\log s) $.
\FullSpace

Overall, $\Psi_{0}$ will be the union of all the above proofs, so that $\Psi_{0}$ contains all equations (\ref{eq:Fv+}), (\ref{eq:Fv*}) (for all nodes $ v $ satisfying (\ref{eq:od ehad})), and all equations (\ref{eq:wFv+}) and (\ref{eq:wFv*}) (for all nodes $ v,w $ satisfying (\ref{eq:nizkarti1}) and (\ref{eq:nizkarti2})).  The proof sequence $ \Psi_0 $ has size $\lambda(s,0)= O(s^4) $ and is and depth $O(\log s)$.

\induction
We wish to construct the proof-sequence $ \Psi_{i+1} $.
\para{Part (I).}
 Let $ v $ be any node in $ F $ such that
\[
         2^i < \deg(v) \le 2^{i+1} .
\]

\case 1
Assume that $ v=v_1+v_2 $. We show how to construct the proof of $ [F_{v}]=[F_{v_1}]+[F_{v_2}]$. Let $ m=2^i $. From the definition of $ [\cd] $ we have:
\begin{align}\label{eq:101}
        [F_v]=[F_{v_1+v_2}]  = \sum_{t\in {\cal B}_m(F_v) } [F_{t_1}]\cd[F_{t_2}]\cd[\partial t (F_{v_1+v_2})] \, .
\end{align}
Since $\deg(v_1)=\deg(v_2)=\deg(v)$, we also have
\begin{align}\label{eq: new1}
        [F_{v_e}]  = \sum_{t\in {\cal B}_m(F_{v_e}) } [F_{t_1}]\cd[F_{t_2}]\cd[\partial t (F_{v_e})], \quad \hbox{for } e\in \{0,1\} \, .
\end{align}

If $t\in {\cal B}_m(F_v)  $ then $ \deg(t)> m = 2^i $. Therefore, for any $t\in {\cal B}_m(F_v)  $, since $ \deg(v)\le 2^{i+1} $, we have $ \deg(v)-\deg(t) <  2^i$ and $ 2\deg(t)>\deg(v) $ and $ t\neq v $ (since $ t $ is a product gate). Thus, by induction hypothesis, the proof-sequence $ \Psi_i $ contains, for any $t\in {\cal B}_m(F_v)  $, the equations
\[
        [\partial t (F_{v_1+v_2})]  = [\partial t F_{v_1}]+[\partial t F_{v_2}] .
\]
Therefore, having $ \Psi_i $ as a premise, we can prove that (\ref{eq:101}) equals:
\begin{equation}\label{eq:201}
\begin{split}
   &       \sum_{t\in {\cal B}_m(F_v) } [F_{t_1}]\cd[F_{t_2}]\cd([\partial t F_{v_1}]+[\partial t F_{v_2}])\\
   & =       \sum_{t\in {\cal B}_m(F_v) } [F_{t_1}]\cd[F_{t_2}]\cd[\partial t F_{v_1}]    +
                                               \sum_{t\in {\cal B}_m(F_v) } [F_{t_1}]\cd[F_{t_2}]\cd[\partial t F_{v_2}].
\end{split}
\end{equation}
If $ t\in {\cal B}_m (F_v) $ and $ t\not\in F_{v_1} $ then $ [\partial t F_{v_1}] = 0 $. Similarly, if $ t\in {\cal B}_m (F_v) $ and $ t\not\in F_{v_2} $ then $ [\partial t F_{v_2}] = 0 $. Hence we can prove
\begin{equation}\label{eq: biggerB}
\sum_{t\in {\cal B}_m(F_v) }[\partial t F_{v_e}]= \sum_{t\in {\cal B}_m(F_{v_{e}}) }[\partial t F_{v_e}],\quad\mbox{for $ e=1,2 $}.
\end{equation}
Thus, using (\ref{eq: new1}) we have that  (\ref{eq:201}) equals:
\begin{equation}\label{eq:421}
\begin{split}
           \sum_{t\in {\cal B}_m(F_{v_1}) } [F_{t_1}]\cd[F_{t_2}]\cd[\partial t F_{v_1}]   +
                                               \sum_{t\in {\cal B}_m(F_{v_2}) } [F_{t_1}]\cd[F_{t_2}]\cd[\partial t F_{v_2}] \\  = [F_{v_1}]+[F_{v_2}].
                                          \end{split}
\end{equation}
The above proof of (\ref{eq:421}) from $ \Psi_i $ has size $ O(s^2\cd m(s,d)) $ and depth $ O(r(s,d)) $.

 \FullSpace
\case 2 Assume that $ v=v_1\cd v_2 $.
We wish to prove $ [F_{v}]=[F_{v_1}]\cd[F_{v_2}]$. Let $ m=2^i $. We assume without loss of generality that $ \deg(v_1)\ge \deg(v_2)$. By the definition of $ [\cd ] $, we have:
\begin{align}
[F_v] = [F_{v_1\cd v_2}] & =  \sum_{t\in{\cal B}_m(F_v)} [F_{t_1}]\cd[F_{t_2}]\cd[\partial t F_v].   \notag
\end{align}
If $ v\in {\cal B}_m(F_v) $, then $ {\cal B}_m = \{v\} $ and we have $\tr{F_{v}}=[F_{v_1}]\cd[F_{v_2}]\cd[\partial_{v} F_v]$. Since $[\partial_{v}F_{v}]=1$, this gives $[F_{v}]= [F_{v_{1}}]\cd[F_{v_{2}}]$, and we are done.

Otherwise, assume $ v\not\in {\cal B}_m (F_v) $. Then   $m=2^i<\deg(v_{1})$ (since, if $ \deg(v_{1})\le m $, then also $ \deg(v_{2})\le m $ and so by definition $ v\in {\cal B}_m(F_v) $). Because, moreover, $\deg(v_1)\leq 2^{i+1}$, we have
\begin{equation} \label{eq: new2}
[F_{v_{1}}]=                                    \sum_{t\in {\cal B}_m(F_{v_1})} [F_{t_1}] \cd [F_{t_2}] \cd [\partial t F_{v_1}]\,.
\end{equation}

Since $ \deg(v) \le 2^{i+1} $ and   $ \deg(t)>m=2^{i} $, for any $ t \in {\cal B}_m(F_v) $, we have
\[
          \deg(v) - \deg(t) \le 2^i \quad\text{ and } \quad 2\deg(t)>\deg(v).
\]
Since $ v\neq t $, by induction hypothesis, $ \Psi_i $ contains, for any $ t\in {\cal B}_m(F_v) $, the equation:
\begin{equation}
 [\partial t (F_{v_1\cd v_2})] = [\partial t F_{v_1}]\cd[F_{v_2}].\label{eq:076}
\end{equation}
Using (\ref{eq:076}) for all $ t \in {\cal B}_m(F_v) $, we can prove the following with a $ \PC(\F) $ proof of size $ O(s^2\cd m(s,d) ) $ and depth $O(r(s,d))$:
\begin{align}
               \sum_{t\in{\cal B}_m(F_v)} [F_{t_1}] \cd [F_{t_2}] \cd [\partial t F_v]
                                & =        \sum_{t\in{\cal B}_m(F_{v})} [F_{t_1}] \cd [F_{t_2}] \cd  [\partial t (F_{v_1\cd v_2})]   \notag\\
                                & =        \sum_{t\in {\cal B}_m(F_{v})} [F_{t_1}] \cd [F_{t_2}] \cd  ([\partial t F_{v_1}]\cd [F_{v_2}])  \notag \\
                                & =  [F_{v_2}]\cd       \sum_{t\in {\cal B}_m(F_{v})}[F_{t_1}] \cd [F_{t_2}] \cd  [\partial t F_{v_1}]. \label{eq:044}
\end{align}
Since ${\cal B}_{m}(F_{v_1})\subseteq {\cal B}_{m}(F_{v})$, we can conclude as in (\ref{eq: biggerB}) that
 \[
                                        \sum_{t\in {\cal B}_m(F_{v})} [F_{t_1}] \cd [F_{t_2}] \cd [\partial t F_{v_1}]  = \sum_{t\in {\cal B}_m(F_{v_1})} [F_{t_1}] \cd [F_{t_2}] \cd [\partial t F_{v_1}]\, .\]
Using (\ref{eq: new2}),
(\ref{eq:044}) equals $[F_{v_{2}}]\cd[F_{v_{1}}]$.
The above proof-sequence (using $ \Psi_i $ as a premise) has size $ O(s^2\cd m(s,d) ) $ and depth $  O(r(s,d))$.
\FullSpace

We now append $ \Psi_{i} $ with all proof-sequences of $ [F_{v}]=[F_{v_1}]+[F_{v_2}]$ for every $ v $ from Case 1, and all proof-sequences of $ [F_{v}]=[F_{v_1}]\cd[F_{v_2}]$ for every $ v $ from Case 2. We obtain a proof-sequence  $ \Psi'_{i+1} $ of size
\[
        \lambda(s,i+1)  \le O(s^3\cd m(s,d))+  \lambda(s,i),
\]
and depth $O(r(s,d)) $.

In Part (II), we extend $ \Psi'_{i+1} $ with more proof-sequences to obtain the final $ \Psi_{i+1} $.

\para{Part (II).}
Let $ v\neq w $ be a pair of nodes in $ F $ such that $ w\in F_v $  and assume that
\[
         2^{i}<\deg(v)-\deg(w) \le 2^{i+1} \text{\ \ and\ \ \ }  2\deg(w)>\deg(v).
\]
Let
\[  m = 2^i + \deg(w) .\]

\case 1
Suppose that $ v=v_1+v_2 $.  We need to prove
\begin{equation}\label{eq:stam baribua}
        [\partial w F_{v}]  = [\partial w F_{v_1}] + [\partial w F_{v_2}]
\end{equation}
based on $ \Psi_i $ as a premise. By construction of $ [\partial w F_{v} ] $,
\begin{align}
        [\partial w F_{v}]  & =       \sum_{t\in{\cal B}_m(F_v)} [\partial t F_v] \cd [\partial w F_{t_1}] \cd [F_{t_2}]\nonumber \\
                                & =       \sum_{t\in{\cal B}_m(F_v)} [\partial t (F_{v_1+v_2})] \cd [\partial w F_{t_1}] \cd [F_{t_2}] . \label{eq:072}
\end{align}
Since $\deg(v_1)=\deg(v_2)=\deg(v)$, we also have
\begin{equation}\label{eq: new3}
[\partial w F_{v_e}]   =       \sum_{t\in{\cal B}_m(F_{v_e})} [\partial t F_{v_e}] \cd [\partial w F_{t_1}] \cd [F_{t_2}],\quad\hbox{ for } e=1,2\,.
\end{equation}
Since $ m = 2^i +\deg(w) $, we have $ \deg(t)>2^i +\deg(w) $, for any $ t\in {\cal B}_m (F_v) $. Thus, by $ \deg(v)-\deg(w) \le 2^{i+1}$, we get that for any $ t\in {\cal B}_m (F_v) $:
\begin{gather*}
         \deg(v)-\deg(t)\le 2^i \text{\ \ \ and \ \ \ }  2\deg(t)>\deg(v),\text{\ \ \ and}
\\  t\neq v \text{\ (since $ t $ is a product gate)}.
\end{gather*}
Therefore, by induction hypothesis, for any $ t\in {\cal B}_m (F_v) $, $ \Psi_i $ contains the equation
 \[ [\partial t (F_{v_1+v_2})] = [\partial t F_{v_1}]+[\partial t F_{v_2}].\]
Thus, based on $ \Psi_i $, we can prove that (\ref{eq:072}) equals:
\begin{align}\notag
                                &        \sum_{t\in{\cal B}_m(F_v)} ([\partial t F_{v_1}]+[\partial t F_{v_2}]) \cd [\partial w F_{t_1}] \cd [F_{t_2}] \\
                         = &        \sum_{t\in{\cal B}_m(F_v)} [\partial t F_{v_1}] \cd  [\partial w F_{t_1}] \cd [F_{t_2}] +
                                                      \sum_{t\in{\cal B}_m(F_v)} [\partial t F_{v_2}] \cd  [\partial w F_{t_1}] \cd [F_{t_2}].\label{eq:235}
\end{align}
As in (\ref{eq: biggerB}), using (\ref{eq: new3}) we can derive the following from (\ref{eq:235}):
\[
\begin{split}
                                       \sum_{t\in{\cal B}_m(F_{v_1})} [\partial t F_{v_1}] \cd  [\partial w F_{t_1}] \cd [F_{t_2}] +
                                                      \sum_{t\in{\cal B}_m(F_{v_2})} [\partial t F_{v_2}] \cd  [\partial w F_{t_1}] \cd [F_{t_2}]
                                                \\
                                                =[\partial w F_{v_1}] + [\partial w F_{v_2}] .
\end{split}
\]
The proof of (\ref{eq:stam baribua}) from $ \Psi_i $ shown above has size $ O(s^2\cd m(s,d)) $ and depth $ O(r(s,d)) $.
\FullSpace


\case 2
Suppose that $v = v_1\cd v_2$. We assume without loss of generality that $ \deg(v_1)\ge\deg(v_2) $ and show how to prove
\begin{equation}\label{eq:stam bashlishit}
[\partial w F_{v}]  = [\partial w F_{v_1}] \cd [F_{v_2}] .
\end{equation}
By construction of $ [\partial w F_{v} ] $:
\begin{align}\notag
        [\partial w F_{v}]  & =       \sum_{t\in{\cal B}_m(F_v)} [\partial t F_v] \cd [\partial w F_{t_1}] \cd [F_{t_2}] \\
                                & =       \sum_{t\in{\cal B}_m(F_v)} [\partial t (F_{v_1 \cd  v_2})] \cd [\partial w F_{t_1}] \cd [F_{t_2}] . \label{eq:224}
\end{align}
Similar to the previous case, for any $ t\in {\cal B}_m (F_v) $ we have
\[ \deg(v)-\deg(t)< 2^i \text{\ \ \ and \ \ \ }  2\deg(t)>\deg(v).\]

If $ v\in {\cal B}_m(F_v) $ then $ {\cal B}_m(F_v) =\{v\}$ and so (\ref{eq:224}) is simply $ \partial v F_{v}  \cd  [\partial w F_{v_1}] \cd  [F_{v_2}] =   [\partial w F_{v_1}] \cd  [F_{v_2}]$ as required. Otherwise, assume that $ v\not\in {\cal B}_{m}(F_{v}) $.
By induction hypothesis, $ \Psi_i $ contains the following equation, for any $t\in {\cal B}_m (F_v) $:
\[
        [\partial t (F_{v_1\cd v_2}) ] = [\partial t F_{v_1}] \cd [F_{v_2}].
\]
Using $ \Psi_i $ as a premise, we can then prove that (\ref{eq:224}) equals:
\begin{multline}\label{eq:056}
                                      \sum_{t\in{\cal B}_m(F_v)} \left([\partial t F_{v_1}] \cd [F_{v_2}]\right) \cd [\partial w F_{t_1}] \cd [F_{t_2}]
                                 = \left(       \sum_{t\in{\cal B}_m(F_v)} [\partial t F_{v_1}] \cd [\partial w F_{t_1}] \cd [F_{t_2}]\right) \cd [F_{v_2}].\ \ \ \
\end{multline}
As in (\ref{eq: biggerB}), we have
$\sum_{t\in{\cal B}_m(F_v)} [\partial t F_{v_1}] \cd [\partial w F_{t_1}] \cd [F_{t_2}]
= \sum_{t\in{\cal B}_m(F_{v_1})} [\partial t F_{v_1}] \cd [\partial w F_{t_1}] \cd [F_{t_2}]$. Also, since $v_{1}\cd v_{2}=v\not \in {\cal B}_{m}(F_{v})$, we have $\deg(v_{1})>m= 2^{i}+\deg(w)$, and so
\begin{equation} \label{eq: new4}
        [\partial w F_{v_1}]    =
                \sum_{t\in{\cal B}_m(F_{v_1})} [\partial t F_{v_1}] \cd [\partial w F_{t_1}] \cd [F_{t_2}]\,.
\end{equation}
Hence by (\ref{eq: new4}),
(\ref{eq:056}) equals
$[\partial_{w}F_{v_{1}}]\cd [F_{v_{2}}]$.

The above proof of (\ref{eq:stam bashlishit}) from $ \Psi_i $ has size $ O(s^2\cd m(s,d)) $ and depth $O(r(s,d)) $.
\FullSpace

We now append $ \Psi'_{i} $ from Part (I) (which also contains $ \Psi_i $) with all proof-sequences of $ [\partial w F_{v}]  = [\partial w F_{v_1}] + [\partial w F_{v_2}] $ in Case 1 and all proof sequences $ [\partial w F_{v}]  = [\partial w F_{v_1}] \cd [F_{v_2}]  $ in Case 2, above.  We obtain the proof-sequence  $ \Psi_{i+1} $ of size
\[
        \lambda(s,i+1)  \le O(s^4\cd m(s,d))+  \lambda(s,i),
\]
and depth $O(r(s,d)) $, as required.

\section{Proofs with division}
In this section, we investigate proofs with divisions (as defined in Section \ref{results: divison}), and prove Theorem \ref{thm: main divisions}.

Let us first turn the reader's attention to some peculiarities of the system $\PI$:

\begin{itemize}\item
 We must be careful not to divide by zero in $\PI$. Hence $\PI$ proofs are \emph{not  closed under substitution}. It may happen that $F(z)=G(z)$ has a $\PI$ proof $S$, $F(0)=G(0)$ is defined \iddosml{(according to  the definition in  Section \ref{results: divison})}, but substituting $z$ by $0$ throughout $S$ is not a correct \iddosml{$\PI$} proof \iddosml{(note that a $\PI$ proof is defined so that every circuit in the proof is defined)}.
\item Whereas $\PI$ is  sound with respect to  polynomial identities, it behaves erratically if one considers \emph{proofs from assumptions}. For example, $\PI$ augmented with the axiom $x^{2}-x= 0$ proves that $1=0$.
\Commentdel{Perhaps this comment is not essential: what you show is just that if you add a false assumption with respect to polynomial-identities you get a contradiction. But maybe it's still interesting...?}

\item Prima facie, it is not clear whether a $\PI$ proof of \iddosml{the} equation $F=G$ can be transformed to a proof of $F=G$ \iddosml{that} contains only the variables contained in $F$ and $G$. See Remark \ref{rem: aux}.
\end{itemize}

In the sequel, we will consider substitution instances of equations we prove in $ \PI $. For instance, we will need to substitute $ 0 $ for some variables in the matrix $ X $, when proving equations involving the circuit $ \Det(X) $, and we have to guarantee that our proofs remain correct \iddosml{$\PI$ proofs} after such a substitution.

There are two \emph{general} ways how to securely handle substitutions in $\PI$ proofs. The first one is to substitute only \emph{algebraically independent elements}: replacing variables $z_{1},\dots, z_{k}$ with circuits $H_{1},\dots, H_{k}$ can never produce an undefined proof, if the circuits compute algebraically independent rational functions. The second way is offered in Corollary \ref{cor: subst}. This corollary allows one to construct a new proof of $F(0)=G(0)$ from the proof of $F(z)=G(z)$.  Note, however, that in Corollary \ref{cor: subst} the new proof will be polynomial only if the syntactic degree of $F$ and $G$ is polynomial.

Since the determinant circuit $ \Det $ has an exponential syntactic degree (see Section \ref{sec: determinant}), the second approach to substitution is not suitable for the $\Det$ identities. The first approach, which substitutes algebraically independent elements, often cannot  be used either, because we need to substitute variables by field elements. Therefore, \iddosml{in some cases we must} simply make sure \iddosml{in an ad hoc manner} that the specific substitutions used do not make the proofs undefined. To this end, we use the following terminology: let $\overline {x}= (x_{1},\dots, x_{k})$ be a list of variables and $U= (U_{1},\dots, U_{k})$ a list of circuits with divisions. We say that a circuit $F(\overline {x})$ with divisions is \emph{defined \iddosml{for}} $\overline {x}=U$, if {no divisions by zero occur in $F(U)$}; likewise, we say that a $\PI$ proof $S$ is defined for $\overline {x}=U$ (or simply defined, if the context is clear),  if every circuit in $S$ is defined for $\overline {x}=U$.

\subsection{Eliminating division gates over large enough fields}
 We first prove Theorem \ref{thm: main divisions} under the assumption that the underlying field $ \F $ is large. To eliminate division gates from proofs, we follow the construction of Strassen \cite{Str73}, in which an inverse gate is replaced by a truncated power series. In order to eliminate division gates over small fields, additional work will be needed (see Section \ref{sec:large}).

Let $F$ be a circuit with divisions. We say that $F$ is a circuit \emph{with simple divisions}, if  for every inverse gate $v^{-1}$ in $F$ the circuit  $F_{v}$ does not contain inverse gates. A size $ s $ circuit with division $F$ can be converted to a size $O(s) $ circuit of the form $F_{1}\cd F_{2}^{-1}$, where $F_{1}, F_{2}$ do not contain inverse gates, as follows.

For every node $v$ introduce two nodes $\Den(v)$ and $\Num(v)$ which will compute the numerator and denominator of the rational function computed by $v$, respectively, as follows:

\begin{enumerate}[(i)]
\item If $v$ is an input node of $F$, let $\Num(v):= v$ and $\Den(v)= 1$.
\item If $v=u^{-1}$, let $\Num (v):= \Den(u)$ and $\Den(v):=\Num(u)$.
\item If $v=u_{1}\cdot u_{2}$, let $\Num(v):=\Num(v_{1})\cdot \Num(v_{2})$ and $\Den(v):=\Den(v_{1})\cdot \Den(v_{2})$.
\item If $v=u_{1}+ u_{2}$, let $\Num(v):=\Num(u_{1})\cdot \Den(u_{2})+ \Num(u_{2})\cdot \Den(u_{1})$ and $\Den(v):=\Den(u_{1})\cdot \Den(u_{2})$.
\end{enumerate}

Let $\Num(F)$ and $\Den(F)$ be the circuits with the output node $\Num(w)$ and $\Den(w)$, respectively, where $w$ is the output node of $F$.  The following lemma will be used in Proposition \ref{prop:division}:

\begin{lemma}\label{lem:Den}
{Let $ \F $ be any field.}
\begin{enumerate}
\item \label{Den - jedna} If $F$ is a size $ s $ circuit with division, then
\[
        F=\Num(F) \cd \Den(F)^{-1}
\]
has a $\PI(\F)$ proof of size $O(s)$.
The proof is defined whenever $F$ is defined.

\item \label{Den - dva}
Let $F,G$ be circuits with division. Assume that $F=G$ has a $\PI(\F)$ proof of size  $s$. Then
$\Num(F)\cd\Den(F)^{-1}= \Num(G)\cd\Den(G)^{-1}$ has a $\PI(\F)$ proof of size $O(s)$ such that every circuit in the proof is a circuit with \emph{simple divisions}.
\end{enumerate}
\end{lemma}

\begin{proof}
Part \ref{Den - jedna} is proved by straightforward induction on the size of $F$ and part \ref{Den - dva} by induction on the number of proof lines. We omit the details.
\end{proof}

\HalfSpace

Let $k$ be a fixed natural number and define $\inv_{k}(1-z)$ to be the circuit
\[\inv_k(1-z):=1+z+ \dots +z^k\,.\]
In other words, $\inv_k(1-z) $ is the first $k+1$ terms of the power series expansion of $1/(1-z)$ at $ z=0 $.

Let $F$ be a division-free circuit and let $a:=\wh{F^{(0)}} $. Assume that $ a\neq 0$, that is, the polynomial computed by $F$ has a nonzero constant term, and let $\Inv_{k}(F)$ denote the circuit
\begin{align*}
        \Inv_{k}(F)&:=a^{-1}\cd \inv_{k}(a^{-1}F) \\
                & = a^{-1}\cd\left(1+ (1-a^{-1}F)+ (1-a^{-1}F)^{2}+\dots+ (1-a^{-1}F)^{k}\right)
                \, .
\end{align*}
Note that $a^{-1}$ is a field element and hence $\Inv_{k}(F)$ is a circuit \emph{without division}. 
{The following lemma shows that $\Inv_k(F) $ can \emph{provably} serve as the inverse polynomial of $ F $ ``up to the $ k $ power'':}

\begin{lemma}\label{lem:inv}
Let $ \F $ be any field and let $F$ be a size $ s $ circuit without division such that $\widehat {F^{(0)}}\not=0$. Then the following have $\PC(\F)$ proofs of size $s\cdot \poly(k)$:
\begin{eqnarray}(F\cdot \Inv_{k} F )^{(0)}&=& 1\\
(F\cdot \Inv_{k} F )^{(i)}&=& 0, \,\, \hbox{for}\,\, 1\leq i\leq k\,.
\end{eqnarray}
\end{lemma}

\begin{proof}
Let $z $ abbreviate the circuit $ 1-a^{-1}F$. Then \iddosml{we can easily prove} $F= a(1-z)$ and \iddosml{by definition}  $\Inv_{k}(F)= a^{-1}(1+z+z^{2}+\dots +z^{k})$.
By elementary rearrangement, we can prove
\[ F\cdot \Inv_{k}(F) =  (1-z)(1+z+z^{2}+\dots z^{k})= 1-z^{k+1}\,.
\]
By Lemma \ref{lem:bounded degree}, $(F\cdot \Inv_{k}(F))^{(0)}= 1 - (z^{k+1})^{(0)}$ and $(F\cdot \Inv_{k}(F))^{(i)}=  (z^{k+1})^{(i)}$, for $i>0$. It is therefore sufficient to prove for every $i\leq k$, $(z^{k+1})^{(i)}=0$. This follows by induction using Lemma \ref{lem:bounded degree} and the fact that $z^{(0)}= 0$.
 \end{proof}

The dependency on the field comes from the following fact, which follows from the Schwartz-Zippel lemma \cite{Sch80,Zip79}:

\begin{fact*} Let $f_{1},\dots , f_{s}\in \F[X] $ be non-zero polynomials of degree $\leq d$, where $X= \{ x_{1},\dots x_{n}\}$. Assume that $|\F|> sd $.
Then there exists $\bar a\in \F^{n}$ such that  $f_{i}(\bar a)\not=0$ for every $i\in \{1,\dots ,s\}$.
\end{fact*}

\begin{proposition} \label{prop:division}
There exists a polynomial $p$ such that the following holds.
Let $F,G$ be circuits without division of syntactic degree at most $d$. Assume that $F=G$ has a $\PI(\F)$ proof with divisions of size at most $s$ and suppose that $\vert \F\vert > 2^{\Omega(s)}$.  Then $F=G$ has a $\PC(\F)$ proof of size $s\cdot p(d)$.
\end{proposition}

\begin{proof}
Let $S$ be a $\PI(\F)$ proof of $F=G$ of size $s$. By Lemma \ref{lem:Den}, we can assume that the proof contains only simple divisions. 
\iddosml{Consider the set \(\mathcal{U}\) of all nodes \(u^{-1}\) occurring in some circuit in \(S\), and let  $\cal C$ be the set of all circuits computed by some node \(u\), for
$u^{-1}\in\mathcal U$}. Then $\vert{\cal C}\vert\leq s$ and $\deg\iddosml{(H)} \leq  2^{\Omega(s)}$ for every $H\in \cal C$, since $H$ has size at most $s$. By the Fact above, there exists a point $b\in \F^{n}$ such that $\wh{H}(b)\not=0$ for every $H\in {\cal C}$, where $n$ is the number of variables in $S$.

Without loss of generality, we can assume that $b= \langle 0,\dots ,0\rangle$ \iddodefer{(otherwise, we can substitute in $S$ the variables with a suitable linear transformation; since the substitution for each variable is algebraically independent, the proof is still defined under this substitution)}.\mardel{I hope this is correct.} Let $S^{\prime}$ be the sequence of equations obtained by replacing every circuit $(H)^{-1}$ in $S$ by $\Inv_{k}(H)$.
\ignore{(That is, remove the inverse gate, and replace it by the output gate of $\Inv_{k}(H)$.)}
The sequence $S^{\prime}$ does not contain divisions, but is not yet a correct proof, since the translation $F\cdot \Inv_{k}( F)= 1$ of the axiom D is not a \iddosml{legal axiom anymore}. However, we claim that for every equation $F_{1}=F_{2}$ in $S^{\prime}$ and every $k\leq d$, $F_{1}^{(k)}=G_{1}^{(k)}$ has a $\PC$ proof of size $s\cdot p(d)$ for a suitable polynomial $p$. The proof is constructed by induction on the length of $S^{\prime}$, as in Proposition \ref{prop: main hom}. The case of the axiom D follows from Lemma \ref{lem:inv}:
$(F\cdot \Inv_{k}( F))^{(0)}=1 = 1^{(0)}$ and $(F\cdot \Inv_{k}( F))^{(j)}=0 = 1^{(j)}$, if $j>0$.
Consequently, we obtain proofs of $F^{(k)}=G^{{(k)}}$, for every $k\leq d$. By Lemma \ref{lem:bounded degree2}, we have $\PC(\F)$ proofs of $F= \sum_{k\leq d} F^{(k)}$, $G= \sum_{k\leq d} G^{(k)}$. This gives $\PC(\F)$ proofs of $F=G$ with the correct size.
\end{proof}

Another application of Schwartz-Zippel lemma \iddosml{we shall need} is the following:

\begin{proposition}\label{rem: aux} Let $\F$ be an arbitrary field and assume that $F=G$ has a $\PI(\F)$ proof of size $s$. Then there exists a $\PI(\F)$ proof of $F=G$ of size $O(s^{2})$ which contains only the variables appearing in $F$ or $G$.
\end{proposition} 

\mardel{I think it should be a Claim or Propositional, if it has a proof.}

\begin{proof} Let $S$ be a proof of $F=G$ of size $s$ which contains variables $z_{1},\dots, z_m$ not appearing in $F$ or $G$.
Assume that $F$ or $G$ actually contain at least one variable $x$, otherwise the statement is clear.
It is sufficient to find a substitution $z_{1}=H_{1},\dots, z_{m}=H_{m}$ for which the proof $S$ is defined and $H_{1},\dots, H_{m}$ are circuits of size $O(s)$ in
the variable $x$ only.
We will choose the substitution from the set $M=\{x^{1},x^{2},x^{3}\dots, x^{2^{cs}}\}$, where $c$ is a sufficiently large constant. Note that $x^{p}$ can be computed by a circuit of size $\log_{2}p+2$, and so every circuit in $M$ has size $O(s)$.
That such a substitution exists can be shown as in Proposition \ref{prop:division}, when we consider $M$ as a subset of the field of rational functions.
\end{proof}

\subsection{Taylor series} \label{sec:Taylor}
For a later application, we need to introduce the basic notion of a power series. Let $F=F(\overline{x},z)$ be a circuit with division. We will define $\Coef_{z^{k}}(F)$ as a circuit
in the variables $\overline {x}$, computing the coefficient of $z^{k}$ in $F$, when $F$ is written as a power series at $ z=0$. This is done as follows:
\HalfSpace

\case 1 Assume first that no division gates in $F$ contain the variable $z$. Then we define $\Delta_{z^{k}}(F)$ by the following rules (the definition is similar to that of $F^{(k)}$ in Section \ref{sec:bounding}, and so we will be less formal here)\mardel{Where is it less formal?}:
\begin{enumerate}[(i)]
\item $\Delta_{z}(z):=1$ and $\Delta_{z^{k}}(z):=0$, if $k>1$.
\item If $F$ does not contain $z$, then $\Delta_{z^{0}}(F):=F$ and $\Delta_{z^{k}}(F):=0$, for $k>0$.
\item $\Delta_{z^{k}}(F+G)= \Delta_{z^{k}}(F)+\Delta_{z^{k}}(G)$.
\item $\Delta_{z^{k}}(F\cdot G)= \sum_{i=0}^k\Delta_{z^{i}}(F)\cd\Delta_{z^{k-i}}(G)$.
\end{enumerate}

\case 2 Assume  that some division gate in $F$ contains $z$. We let:
\[F_{0}:= \left ((\Den(F))(z/0)\right)^{\sharp}\,,\]
where, given a circuit $ G $,  $G^{\sharp}$ is the non-redundant version of $G$ (see definition in Section \ref{sec:bounding}) and $ G(z/0) $ is obtained by substituting  in $ G $ all occurrences of $ z $ by the constant $ 0 $.
\mardel{I don't understand where and why precisely you need to assume non-redundancy of circuits...?}
In case $\wh{F_{0}}\not=0$, we define:
\[
        \Delta_{z^{k}}(F):= F_{0}^{-1} \cdot \Delta_{z^{k}} \left(\Num(F)\cd\inv_{k}\left(F_{0}^{-1}\cd\Den(F)  \right)\right)\,.
\]
Note that $ z $ does not occur in any division gate inside $ \Num(F)\cd\inv_{k}\left(F_{0}^{-1}\cd\Den(F)  \right) $, and so $ \Delta_{z^{k}}(F) $ is well-defined.

\Commentdel{I feel the definition of $ \Delta_{z^{k}}(F)$ (for Case 2) is quite confusing. As far as I checked \[ \Delta_{z^{k}}(F):= F_{0}^{-1} \cdot \Delta_{z^{k}} \left(\Num(F)\cd F_0\cd\Inv_{k}\left(\Den(F)  \right)\right),\] so maybe this should be stated as such. And also, is it possible to simply define: 
 \[ \Delta_{z^{k}}(F):= \Delta_{z^{k}} \left(\Num(F)\cd \Inv_{k}\left(\Den(F)  \right)\right ) ~~? \]
 And if not, why? \\ \\\ Also, if F contains division gates, but no division gate in F contains z then $\Delta_{z^k}(F)$ may contain division gates. But if some division gate in F contains z then $\Delta_{z^k}(F)$ does not contain any division gate.     
}

We summarize the main properties of $\Coef_{z^{k}}$ as follows:

\begin{proposition}\label{prop:Taylor}
\
\begin{enumerate}
\item \new{If $F$ is a circuit without division of syntactic degree at most $d$ and size $
s$ then $F= \sum_{i=0}^d \Coef_{z^{i}}(F)\cdot  z^{i}$
has a $\PC$ proof of size $s\cdot\poly(d)$.}
\item \new{
If $F_{0},\dots, F_{k}$ are circuits with divisions not containing the variable $z$, then
$\Coef_{z^{j}}\left(\sum_{i=0}^{k}F_{i} z^{i}\right)= F_{j}$
has a polynomial size $\PI$ proof, for every $j\leq k$.}
\item
Assume that $F,G$ are circuits with divisions such that $F=G$ has a $\PI$ proof of size $s$ \iddosml{that} is defined for $z=0$. Then
\[ \Delta_{z^k}(F)= \Delta_{z^{k}}(G)\]
has a $\PI$ proof of size $s\cdot \poly(k)$.
\end{enumerate}
\end{proposition}
  The proofs are almost identical to those of Proposition \ref{prop: main hom}  and Proposition \ref{prop:division}. We omit the details.

\section{Simulating large fields in small ones}\label{sec:large}
\new{Recall the notation on matrices given in Section \ref{matrix notation}. Mainly,
 matrices are understood as matrices whose entries are circuits and operations on matrices are operations on circuits. }

\begin{lemma}\label{lem:matrices}\label{lem: matrices} Let $X,Y,Z$ be $n\times n$ matrices of distinct variables and $I_{n}$ the identity matrix. Then the following identities have polynomial-size $\PC(\F)$ proofs:
\begin{tabbing}
\qquad\qquad  \= $ X+Y=Y+X  $ ~~~~~~~~~~~~~~~~~~~~~~~ \= $ X+(Y+Z)= (X+Y)+Z$ \\
  \> $ X\cdot(Y+Z)= X\cdot Y+ X\cdot Z $ \> $ (Y+Z)\cdot X= Y\cdot X+ Z\cdot X $\\   
 \qquad \qquad\>$ X\cdot (Y\cdot Z)= (X\cdot Y)\cdot Z $ \> \new{$X\cdot I_{n}=I_{n}\cdot X = X.$}
\end{tabbing}
\new{Similarly for non-square matrices of appropriate dimension.}
\end{lemma}

\begin{proof} Each of the equalities is a set of $n^{2}$ correct equations with degree $\leq 3$ and size $O(n)$. Every such equation has a $\PC$-proof of size $O(n^{3})$.
\end{proof}

Let $\F_{1}= GF(p)$ and $\F_{2}= GF(p^{n})$, where $p$ is a prime power. \mardel{$\F_2$ sometimes denotes $GF(2)$, so the notation is a bit confusing...}We will show how to simulate proofs in $\PC(\F_{2})$ by proofs in $\PC(\F_{1})$.
Recall  that $\F_{2}$ can be represented by $n\times n$ matrices with elements from $\F_{1}$, that is, there is an isomorphism $\theta$ between $\F_{2}$ and a subset of $GL_{n}(\F_{1})$. \new{We can also assume that $\theta(a)= aI_{n}$ if $a\in \F_{1}\subseteq \F_{2}$.}
 This allows one to treat a polynomial $f$ over $\F_{2}$ as a matrix of $n^{2}$ polynomials over $\F_{1}$. Similarly, we can define a translation of circuits: let $F$ be a circuit with coefficients from  $\F_{2}$. Let $\overline{F}$ be an $n\times n$ matrix of circuits $\{\overline{F}_{ij}\}, \, i,j\in [n]$ with coefficients from $\F_{1}$, defined as follows: for every gate $u$ in $F$, introduce $n^{2}$ gates $\bar{u}= \{\bar{u}_{ij}\}_{i,j\in [n]}$, and let:
\begin{enumerate}
\item If $u\in \F_{2}$ is a constant, let $\bar{u}:= \theta(u)$.
\item If $u$ is a variable, let $\bar{u}:= u\cdot I_{n}$.
\item If $u= v+w$, \new{let} $\bar{u}:=\bar{v}+\bar{w}$, and if  $u= v\cdot w$, let $\bar{u}: =\bar{v}\cdot \bar{w}$
\end{enumerate}
Then $\overline{F}$ is the matrix computed by $\bar{w}$ where $w$ is the output of $F$.

Here, $\bar{v}+ \bar{w}$,  $(\bar{v}\cdot \bar{w})$ and $u\cdot I_{n}$ are understood as the corresponding matrix operations on circuit nodes.

\begin{lemma}\label{lem: large field}
Let $F,G$ be circuits of size $\leq s$ with coefficients from $\F_{2}$. Then
\begin{eqnarray} \overline{F\cplus G}&=& \overline{F}+\overline{G}\,, \ \ \ \,\,\overline{F\ctimes G}= \new{\overline{F}\cdot\overline{G}}\,, \label{large:1}\\
\overline{F}\cdot \overline{G}&=&\overline{G}\cdot \overline{F} \label{large:2}
\end{eqnarray}
have $\PC(\F_{1})$ proofs of size $s\cdot\poly(n)$
\end{lemma}

\begin{proof} Identities (\ref{large:1}) follow from the definition of $\overline{F}$ by means of axioms C1, C2.

Identity (\ref{large:2}) follows by induction on the circuit sizes of $F$ and $G$. We first need to construct the proof of
\[
\overline{z_{1}}\cdot \overline{z_{2}}= \overline{z_{2}}\cdot\overline{z_{1}}\,,
\]
where each $z_{1},z_{2}$ is either a  variable or an element of $\F_{2}$. So assume that $z_{1}$ is a variable. Then $\overline{z_{1}}= z_{1}\cdot I_{n}$.
This gives $\overline{z_{1}}\cdot \overline{z_{2}}= z_{1}\cdot \overline{z_{2}}$. But  $\overline{z_{2}}$ is a matrix  for which each entry commutes with $z_{1}$, which gives a proof of $z_{1}\cdot \overline{z_{2}}= \overline{z_{2}}\cdot z_{1}= \overline{z_{2}}\cdot\overline{z_{1}}$. The case of $z_{2}$ being a variable is similar. If both $z_{1}, z_{2}\in \F_{2}$, we are supposed to prove
$\theta(z_{1})\cdot \theta(z_{2})= \theta(z_{2})\cdot \theta(z_{1}) $. But this is a set of $n^{2}$ true equations of size $O(n)$ which contain only elements of $\F_{1}$, and hence it has a proof of size $O(n^{3})$.
 In the inductive step, use \new{(\ref{large:1}) and} Lemma \ref{lem: matrices}  to construct proofs of $(\overline{F_{1}}+ \overline{F_{2}})\cdot \overline{G}=\overline{G}(\overline{F_{1}}+ \overline{F_{2}}) $ and  of $(\overline{F_{1}}\cdot \overline{F_{2}})\cdot \overline{G}=\overline{G}(\overline{F_{1}}\cdot \overline{F_{2}}) $ from the proofs of $\overline{F_{1}}\cdot \overline{G}=\overline{G}\cdot \overline{F_{1}}$ and $\overline{F_{2}}\cdot \overline{G}=\overline{G}\cdot \overline{F_{2}}$.
\end{proof}
\medskip

\mardel{I think it would be good to restate theorems 10,9 here, for the sake of readability.}
We are now ready to prove Theorem \ref{main: large field}, restated below for the sake of convenience:
\medskip

\ind\textbf{Theorem \ref{main: large field}.}\textit{ Let $p$ be a prime power and $n$ a natural number and let $F,G$ be circuits over $GF(p)$. Assume that $F=G$ has a $\PC(GF(p^{n}))$ proof of size $s$. Then $F=G$ has a $\PC(GF(p))$ proof of size $s\cdot \poly(n)$.}
\medskip

\begin{proof}[Proof of Theorem \ref{main: large field}] Let $F,G$ be circuits with coefficients from $\F_{2}$ such that
 $F=G$ has a $\PC(\F_{2})$ proof of size  $s$. We wish to show that $\overline{F}=\overline{G}$ have proofs of size $s\cdot\poly(n)$ in $\PC(\F_{1})$.
This implies Theorem \ref{main: large field}, for if $F,G$ contain only coefficients from $\F_{1}$ then $\overline{F}_{11}=F$ and $\overline{G}_{11}=G$.

The proof is constructed by induction on the number of lines.
Axioms C1, C2 follow from equations (\ref{large:1}) in Lemma \ref{lem: large field},  and  A4 from equation (\ref{large:2}). A9 is a set of $n^{2}$ true constant equations. The rest of the axioms are application of Lemma \ref{lem: matrices}.  The rules R1, R2 are immediate, and R3, R4 are given by Lemma \ref{lem: large field}.
\end{proof}

Now we can also prove Theorem \ref{thm: main divisions}:
\medskip

\ind\textbf{Theorem \ref{thm: main divisions}.}
\textit{Let $ \F $ be any field and assume that $F $ and $ G$ are circuits without division gates such that $\deg F, \deg G\leq d$. Suppose that $F=G$ has a $\PI(\F)$ proof of size $s$. Then $F=G$ has a $\PC(\F)$ proof of size $s\cdot\poly(d)$.}
\medskip

\begin{proof}[Proof of Theorem \ref{thm: main divisions}] Follows from Theorem \ref{main: large field} and Proposition \ref{prop:division}.
\end{proof}

For a circuit with division $F$, define its syntactic degree by
\[\deg F:= \deg (\Num( F))+\deg (\Den (F)).\]

\begin{corollary}\label{cor: subst}
Let $ \F $ be any field and let $F$, $G$, $H$ be circuits with divisions. Assume that  $\deg  (F) $ and $ \deg (G) $ \iddosml{are} at most $ d$ and that $H$ has size $ s_{1}$.
Suppose that   $F=G$  has a $\PI(\F) $ proof of size $s_{2}$ and that $F(z/H), \,G( z/H) $ are defined. Then $ F(z/ H)=G( z/H) $ has a $\PI(\F)$ proof of size $\,s_{1}s_{2}\cd\poly(d)$.
\end{corollary}

\begin{proof} We aim to construct a proof of $F=G$ of size $s_{2}\cdot \poly(d)$ such that the proof is defined for $z=H$. We can then substitute $H$ for $z$ throughout the proof to obtain a proof of $F(z/H)=G(z/H)$ {of the required size}.
By Lemma \ref{lem:Den}, we have proofs of
\begin{equation}F=\Num(F)\cdot \Den(F)^{-1}\,~~~\mez G=\Num(G)\cdot \Den(G)^{-1}\,.\label{eq: Num}
\end{equation} This and $F=G$ gives a $\PI(\F)$ proof
 of  \[\Num(F    )\cdot \Den(G    )= \Num(G    )\cdot Den(F    )\,,\]
of size {$O(s_{2})$}.
The last equation does not contain division gates, and so it has a $\PC(\F)$ proof of size ${s_{2} \cd\poly(d)}$ by Theorem \ref{thm: main divisions}. This proof is defined for $z=H$ because it does not contain division gates. By Lemma \ref{lem:Den}, the proofs of (\ref{eq: Num}) are defined for $z=H$ {(because $F(z/H) $ and $  G(z/H) $ are defined by assumption)}. In particular, both $\Den(F    )(z/H)$ and $\Den(G    )(z/H)$ are nonzero, and we have a proof of
\[\Num(F    )\cdot \Den(F    )^{-1}= \Num(G    )\cdot \Den(G    )^{-1}\]
 which is defined for $z=H$. Using (\ref{eq: Num}) we obtain a proof of $F=G$ of size ${s_{2}\cd\poly(d)}$ which is defined for $z=H$.
\end{proof}

\section{Computing the determinant}\label{sec: determinant}

We are now done proving the structural properties of $\PC$ and $\PF$ and we proceed to construct proofs of the properties of the determinant.
We first compute the determinant as a rational function.

\subsection{The determinant as a rational function}

\subsubsection*{The definition of $X^{-1}$ and $\Det(X)$}

Let $X= \{x_{ij}\}_{i,j\in [n]}$ be a matrix consisting of $n^{2}$ distinct variables. Recursively, we define \iddosml{an} $n\times n$ matrix $X^{-1}$ whose entries are circuits with divisions.
\begin{enumerate}
\item If $n=1$, let $X^{-1}:=(x_{11}^{-1})$.
\item If $n>1$, partition $X$ as \iddosml{follows:}
\begin{equation}
X= \left(
                \begin{array}{l r}
                        X_{1} & v_{1}^{t}\\
                        v_{2} & x_{nn}
                \end{array}
        \right)\,,
\label{eq: X}
\end{equation}
where $X_{1}=\{x_{ij}\}_{i,j\in [n-1]}$, $v_{1}=(x_{1n},\dots, x_{(n-1)\iddosml{n}})$ and $v_{2}=(x_{n1},\dots, x_{n(n-1)}) $.
Assuming we have constructed $X_{1}^{-1}$, let
\begin{equation}\label{eq: dd}
\dd(X):= x_{nn}- v_{2}X_{1}^{-1}v_{1}^{t}\,.
\end{equation}
$\dd(X)$ computes a single non-zero rational function and so $\dd(X)^{-1}$ is defined.
Finally, let
\begin{equation}
X^{-1}:= \left( \begin{array}{l r}
X_{1}^{-1}(I_{n-1}+\dd(X)^{-1}v_{1}^{t} v_{2 }X_{1}^{-1}) & -\dd(X)^{-1}X_{1}^{-1}v_{1}^{t} \\
-\dd(X)^{-1} v_{2} X_{1}^{-1} & \dd(X)^{-1}
\end{array}\right)\,.
\label{eq:inverse}
\end{equation}
\end{enumerate}

\iddosml{The circuit} $\Det(X)$ is defined as \iddosml{follows}:
\begin{enumerate}\item If $n=1$, let $\Det(X):= x_{11}$.
\item If $n>1$, partition  $X$ as in (\ref{eq: X}) and let $\dd(X)$ be as in (\ref{eq: dd}).
Let \[\Det(X):= \Det(X_{1})\cdot \dd(X)= \Det(X_1)\iddo{\cd}(x_{nn}- v_{2}X_{1}^{-1}v_{1}^{t})\,.\]
\end{enumerate}

The definition in (\ref{eq:inverse}) should be understood as a circuit with $n^{2}$ outputs which takes $X_{1}^{-1}, v_{1}, v_{2}, x_{nn} $ as inputs and moreover, such that the inputs from $X_{1}^{-1}$ occur exactly once (so we slightly deviate from earlier notation).
Altogether,  we obtain polynomial size circuits for $X^{-1}$ and $\Det(X)$.
The fact that $ \Det(X) $ indeed computes the determinant (as a rational function) is a consequence \iddosml{of} Proposition \ref{lem:det} below, where we show that $ \PI $ can prove the two identities which characterize the determinant. That $X^{-1}$ computes the matrix inverse is proved in Proposition \ref{prop: inverse}. \mardel{"circuits with division" or "circuits with divisions"? Should be uniform throughout  paper.}

It should be emphasized that both  $X^{-1}$ and $\Det(X)$ are circuits with \iddosml{division} and hence not always defined when substituting for $X$.
Let $A:=\{a_{ij}\}_{i,j\in [n]}$ be an $n\times n$ matrix whose entries are circuits with division. We will say that $A$ is \emph{invertible} if the circuit $A^{-1}$ is defined---that is, when we substitute the entries of $A$ into $X^{-1}$,  the circuit does not use divisions by zero.
Note that $A^{-1}$ may be undefined even if $A$ has inverse ``in the real world''.
For example, if \[A=\left(\begin{array}{c c} 0 & 1\\ 1 & 0\end{array}\right)\]
then both $A^{-1}$ and $\Det(A)$ are undefined, and so $A$ is  not invertible in our sense.
Moreover, note that $\Det(X)$ has an \emph{exponential} syntactic degree which, in view of Corollary \ref{cor: subst}, further obscures \iddosml{the possibility to apply substitutions} in $\PI$ proofs.

On the other hand, let us state the basic cases when the determinant and matrix inverse are defined. \iddosml{Setting  $$A[k]:=\{a_{ij}\}_{i,j\in [k]},$$
we have the following: }\mardel{I put this notation here because it's used often in what follows, and so it might help the reader to emphasize this notation.}
\begin{enumerate}
\item If $A$ is invertible (meaning the circuit $ A^{-1} $ is defined) then $\Det(A)$ is defined.
        
\item If the entries of $A$ compute algebraically independent rational functions then \iddosml{$A$ is invertible}.
        
\item If $A$ is a triangular matrix with $a_{11},\dots, a_{nn}$ on the diagonal such that $a_{11}^{-1},\dots, a_{nn}^{-1} $ are defined then $A$ is invertible.
        
\item The matrix $A$ is invertible if and only if $A[1],\dots, A[n-1]$ are invertible and $\dd(A)^{-1}$ is defined.
\end{enumerate}

\subsection*{Properties of matrix inverse}

\begin{proposition}\label{prop: inverse}
Let $X= \{x_{ij}\}_{i,j\in [n]}$ be a matrix with $n^{2}$ distinct variables. Then both
\[X\cdot X^{-1} = I_{n} \ \ \ \hbox{and}\ \ \  X^{-1}\cdot X= I_{n}\]
have a polynomial-size $\PI$ proof. The proof is defined for $X= A$, whenever $A$ is invertible.
\end{proposition}

\begin{proof} Let us construct the proofs of $X\cdot X^{-1}=I_{n}$ and $ X^{-1}\cdot X= I_{n}$ by induction on $n$. If $n=1$, we have $x_{11}\cdot
x^{-1}_{11}=x^{-1}_{11}\cdot x_{11}= 1$ which is a $\PI$ axiom.  Otherwise let $n>1$ and $X$ be as in (\ref{eq: X}).
We want to construct a polynomial size proof of $X\cdot X^{-1}=I_{n}$ from the assumption $X_{1}X_{1}^{-1}=I_{n-1}$. This implies that $X\cdot X^{-1}=I_{n}$ has a polynomial size proof.

For brevity, let $a:= \dd(X)$.
 Using some rearrangements, and the definition of $a$, we have:
\begin{align*}
X\cdot X^{{-1}}&=
        \left(
                \begin{array}{l r} X_{1} & v_{1}^{t}\\
                                                                v_{2} & x_{nn}
                \end{array}
        \right)\cdot
        \left(
                \begin{array}{l r}                      
X_{1}^{-1}(I_{n-1}+a^{-1}v_{1}^{t} v_{2 }X_{1}^{-1}) & ~~~~-a^{-1}X_{1}^{-1}v_{1}^{t} \\
-a^{-1} v_{2} X_{1}^{-1} & a^{-1}
                \end{array}
        \right)\\
        &=\left(
                        \begin{array}{l r}
                                I_{n-1}+a^{-1}v_{1}^{t} v_{2} X_{1}^{-1} - a^{-1}v_{1}^{t} v_2     X_{1}^{-1}& - a^{-1}v_1^t + a^{-1}v_1^t     \\
                                v_2    X_{1}^{-1}+ a^{-1}(v_2    X_{1}^{-1}v_1^t       -x_{nn} )v_2     X_{1}^{-1} &
                                                ~~~~~a^{-1}(-v_2    X_{1}^{-1}v_1^t     + x_{nn} )
                        \end{array}
                \right) \\
        &=\left(
                        \begin{array}{l r}
                                I_{n-1}& 0\\
                                v_2    X_{1}^{-1}-  a^{-1}a v_2     X_{1}^{-1} & ~~~a^{-1}a
                        \end{array}
                \right) \\
        &=  \left(
                                \begin{array}{l r}
                                        I_{n-1}& 0\\
                                        0 & 1
                                \end{array}
                        \right).
\end{align*}
Here we use the fact that basic properties of matrix addition and multiplication have feasible proofs (see Lemma \ref{lem:matrices}).

The proof of $X^{-1}\cdot X= I_{n}$ is constructed in a similar fashion (where we use the assumption $X_{1}^{-1}X_{1}=I_{n-1}$ instead).
Moreover, if $A$ is an $n\times n$ matrix such that $A^{-1}$ is defined, the proofs of $A\cdot A^{-1}=A^{-1}\cdot A=I_{n}$ are defined. (This is because they employ only the inverse gates appearing already in the definition of $X^{-1}$.)
\end{proof}

\begin{corollary}\label{cor:invers}
The identity $(XY)^{-1}= Y^{-1}X^{-1}$ has a polynomial-size proof in $ \PI $.
The proof is defined for $X=A, Y=B$ whenever $A, B$ and $ AB$ are invertible.
\end{corollary}

Beware that  invertibility of $A$ and $B$ does not guarantee invertibility of $AB$.

\begin{proof}Let $Z:= (XY)^{-1}$. Then  $(Z(XY))Y^{-1}X^{-1}=Y^{-1}X^{-1}$. On the other hand,
$(Z(XY))Y^{-1}X^{-1}= Z(X(YY^{-1}) X^{-1}= Z$ and so $Z=Y^{-1}X^{-1}$.
\end{proof}

An application of  \iddosml{Corollary \ref{cor:invers}} is the following technical observation.
Let $X$ be as in (\ref{eq: X}) and similarly
$Y=
\left( \begin{array}{l r} Y_{1} & u_{1}^{t}\\
u_{2} & y_{nn} \end{array}\right)\,.
$
Comparing the entries in the bottom right corners of $(XY)^{-1}$ and $Y^{-1}X^{-1}$, we obtain that
\begin{equation}\label{eq: ddXY}
\dd(Y)\dd(X)= \dd(XY)(1+ u_{2}Y_{1}^{-1} X_{1}^{-1} v_{1}^{t})\,,
\end{equation}
 has a polynomial size $\PI$ proof
(the proof is defined for \iddosml{$X=A$ and $Y=B$} whenever $A$, $B$ and $AB$ are invertible).

It is often easier to argue about triangular matrices. We summarize their useful properties in \iddosml{what follows}:
\begin{proposition}\label{lem: LU}
\begin{enumerate} \item Let $A, L,U$ be $n\times n$ matrices with $L$ lower triangular and $U$ upper triangular.
If $A,L,U$ are invertible then so are $LA$ and $AU$. \label{LU1}
\item
 Let $X$ be an $n\times n$ matrix of  distinct variables. Then there exists a lower triangular matrix $L(X)$ and an upper triangular matrix $U(X)$ such that
$X=L(X)\cdot U(X)$
has a polynomial size $\PI$ proof. 
\iddosml{If} $A$ is invertible, \iddosml{then} the proof is defined for $X=A $, \iddosml{and  also} $L(A), U(A)$ are invertible. \mardel{You mean to say ``(which imply also that $L(A), U(A)$ 
are invertible)."?\\ Check if the change is correct.}
\label{LU2}
\end{enumerate}
\end{proposition}

\begin{proof}
Part \ref{LU1} follows from the fact that $(LA)[k]=L[k]A[k]$ and $\dd((LA)[k])=\dd( L[k])\dd(A[k])$ for every $k\in \{1,\dots,n\}$ \Commentdel{I couldn't verify this}.
\iddosml{And similarly for $AU$.}

In part \ref{LU2}, the matrices $L(X), U(X)$, as well as the $\PI$ proof,  are constructed by induction on $n$. If $n=1$, let $L(x_{11})= x_{11}$ and $U(x_{11})= 1$. If $n>1$, write $X$ as in (\ref{eq: X}). Assuming we have $X_{1}=L(X_{1})U(X_{1})$, we have
\[\left(
                \begin{array}{l r}
                        X_{1} & v_{1}^{t}\\
                        v_{2} & x_{nn}
                \end{array}
        \right)
        = \left(
                \begin{array}{l r}
                        L(X_{1}) & 0\\
                        v_{2}U(X_{1})^{-1} & ~~~~~x_{nn} - v_{2}X_{1}^{-1}v_{1}^{t}
                \end{array}
        \right)
\cdot
        \left(
                \begin{array}{l r}
                        U(X_{1}) & ~~~~L(X_{1})^{-1} v_{1}^{t}\\
                        0 & 1
                \end{array}
        \right)\,.
        \]
Verifying that the proof is defined for an invertible $A$, and that $L(A), U(A)$ are invertible, is straightforward.
\end{proof}

\subsubsection*{Properties of $\Det$}
We now want to prove Proposition \ref{lem:det} which is a $\PI$ \iddosml{analogue} of Theorem \ref{thm: main det}. We first prove the following lemma:

\begin{lemma}\label{lem: vv} Let  $A$ be an invertible $n\times n$ matrix and let $v_{1}, v_{2}$ be $n\times 1$ vectors such that $A+v_{1}^{t}v_{2}$ is invertible. Then
\begin{equation}\label{eq: vv1}\Det(A+v_{1}^{t}v_{2})= \Det(A)(1+v_{2}A^{-1}v_{1}^{t})\end{equation}
has a polynomial size $\PI$ proof.
\end{lemma}
\Commentdel{***I couldn't verify all steps in this proof. I wrote which places precisely.}

\begin{proof}
The proof is constructed by induction on $n$. If $n=1$, the identity is immediate. If $n>1$, partition
$A$ and $A+v_{1}^{t}v_{2}$ as in (\ref{eq: X}), i.e.,
 \[A= \left(\begin{array}{l r}
                        A_{1} & w_{1}^{t}\\
                        w_{2} & a_{nn}
                \end{array}\right)\,\hbox{ and }\,\,
A+v_{1}^{t}v_{2}= \left(\begin{array}{l r}
                        A_{1}+u_{1}^{t}u_{2} & w_{1}^{t}+c_{2}u_{1}^{t}\\
                        w_{2}+c_{1}u_{2} & a_{nn}+c_{1}c_{2}
                \end{array}\right)\,,           
                \]
where $v_{1}=(u_{1}, c_{1})$ and $v_{2}=(u_{2},c_{2})$. We want to construct a polynomial size proof of (\ref{eq: vv1})
from the assumption $\Det(A_{1}+u_{1}^{t}u_{2})= \Det(A_{1})(1+u_{2}A_{1}^{-1}u_{1}^{t})$. This implies that (\ref{eq: vv1}) has a polynomial size proof.

By the definition of $\Det$, we have
\[\Det(A)=\Det(A_{1})\dd(A)\, \,\, \hbox{ and ~}\,\, \Det(A+v_{1}^{t}v_{2})= \Det(A_{1}+u_{1}^{t}u_{2})\dd(A+v_{1}^{t}v_{2})\,.\]
By the assumption, $\Det(A_{1}+u_{1}^{t}u_{2})= \Det(A_{1})(1+u_{2}A_{1}^{-1}u_{1}^{t})$ and so (\ref{eq: vv1}) is equivalent to
\[\Det(A_{1})(1+u_{2}A_{1}^{-1}u_{1}^{t})\dd(A+v_{1}^{t}v_{2})= \Det(A_{1})\dd(A)(1+v_{2}A^{-1}v_{1}^{t})\,. \]
Hence in order to prove (\ref{eq: vv1}), it is sufficient to prove
\begin{equation}\label{eq: vv2}
(1+u_{2}A_{1}^{-1}u_{1}^{t})\dd(A+v_{1}^{t}v_{2})= \dd(A)(1+v_{2}A^{-1}v_{1}^{t})\,.
\end{equation}

In order to prove (\ref{eq: vv2}), we first prove its special case
\Commentdel{Perhaps it will be easier to read if you put this special case into a separate claim}\begin{equation}\label{eq: vv3}
(1+\bar u_{2}\bar u_{1}^{t})\dd(I_{n}+\bar v_{1}^{t}\bar v_{2})= (1+\bar v_{2}\bar v_{1}^{t})\,\iddo{,}
\end{equation}
where $\bar v_{1}=(\bar u_{1}, \bar c_{1})$ and $\bar v_{2}=(\bar u_{2}, \bar c_{2})$ are vectors such that
$I_{n}+\bar v_{1}^{t}\bar v_{2}$ is invertible.
Let $\alpha:= \bar u_{2}\bar u_{1}^{t}$.
By the definition of $\dd$
\[
\dd(I_{n}+\bar v_{1}^{t}\bar v_{2})= 
1+\bar c_{1} \bar c_{2} - \bar c_{1} \bar c_{2}\bar u_{2}(I_{n-1}+\bar u_{1}^{t}\bar u_{2})^{-1}\bar u_{1}^{t} 
\]
and it is \iddosml{also easy to derive:}  \Commentdel{I couldn't verify this? If it stems from the definition of $X^{-1}$ and proved by induction on $n$, then I guess it should be proved. Or maybe there's a direct proof of this?}
\[
(I_{n-1}+\bar u_{1}^{t}\bar u_{2})^{-1}= I_{n-1}- (1+\alpha)^{-1}\bar u_{1}^{t}\bar u_{2}\,.
\]
Hence we obtain
\begin{align*}(1+\bar u_{2}\bar u_{1}^{t})\dd(I_{n}+\bar v_{1}^{t}\bar v_{2})=& (1+\alpha)(1+\bar c_{1}\bar c_{2}-\bar  c_{1}\bar c_{2} \bar u_{2}(I_{n-1}-(1+\alpha)^{-1}\bar u_{1}^{t}\bar u_{2})\bar u_{1}^{t}\iddo{)}\,\\
=& (1+\alpha)\left(1+\bar c_{1}\bar c_{2} - \bar c_{1}\bar c_{2}\bar u_{2}\bar u_{1}^{t} - \bar c_{1}\bar c_{2}(1+\alpha)^{-1}(\bar u_{2}\bar u_{1}^{t})^{2}\right)\\
=& (1+\alpha)(1+\bar c_{1}\bar c_{2} - \bar c_{1}\bar c_{2}\alpha - \bar c_{1}c_{2}(1+\alpha)^{-1}\alpha^{2})\\ =& 1+\bar c_{1}\bar c_{2}+\alpha\\
= &1+ \bar v_{2}\bar v_{1}^{t}\,,
\end{align*}
which proves (\ref{eq: vv3}).

In order to conclude (\ref{eq: vv2}), let
$L:=L(A)$ and $U:=U(A)$ be the
matrices from Proposition \ref{lem: LU}. That is, $L$ and $U$ are \iddosml{invertible lower and upper triangular matrices, respectively,  so that $A=LU\,$  has a polynomial-size proof, and hence also $A^{-1}=U^{-1}L^{-1}$ has a polynomial-size proof by Corollary \ref{cor:invers}.} \mardel{change phrase only}

Let
\[\bar v_{1}^{t}:= L^{-1}v_{1}^{t}\,\hbox{ ~~~and~~~ }\,\, \bar v_{2}:= v_{2}U^{-1}\,.\]
The definition guarantees that \Commentdel{Why precisely?}
\begin{align}
&\bar u_{2} \bar u_{1}^{t}= u_{2}A^{-1}_{1}u_{1}^{t}\,\mbox{~~~~\iddosml{and}~~~~}\,\,\bar v_{2} \bar v_{1}^{t}= v_{2}A^{-1}v_{1}^{t}\label{eq: vv4}
\end{align}
have polynomial size proofs, \iddosml{where $\bar v_1=(\bar u_1, c_1)$ and $\bar v_2=(\bar u_2,\bar c_2)$}.\mardel{I defined $\bar v_1$..., again, cause they're, formally,  different from above...}  Moreover, $A+v_1^tv_2= L(I_{n}+\bar v_{1}^t\bar v_{2})U$, which also shows that $I_{n}+\bar v_{1}^t\bar v_{2}$ is invertible \Commentdel{Why is this invertible?}. \iddosml{Equation} (\ref{eq: ddXY}) implies that $\dd(LB)=\dd(L)\dd(B)$ and
$\dd(BU)=\dd(B)\dd(U)$  have \iddosml{polynomial-size proof} (for any invertible $B$) \Commentdel{I also didn't get why this holds..?}.
Hence
\[
\dd(A+v_1^{t}v_{2})=\dd(L)\dd(U)\dd(I_{n}+\bar v_{1}^{t}\bar v_{2})= \dd(A)\dd(I_{n}+\bar v_{1}^{t}\bar v_{2})\,.
\]
This, \iddosml{together with} (\ref{eq: vv4}), gives (\ref{eq: vv2}) from (\ref{eq: vv3}). \Commentdel{Why?}
\end{proof}

\begin{proposition}\label{lem:det}\mbox{}
\begin{enumerate}
\item Let $U$ be an (upper or lower) triangular matrix with $u_{1},\dots ,u_{n}$ on the diagonal. If $u_{1}^{-1},\dots, u_{n}^{-1}$ are defined then \iddosml{the following  has a polynomial-size $\PI$ proof:}
\[\Det(U)= u_{1}\cdots u_{n}\,.\]

\label{part1}
\item
 Let $X$ and $Y$ be $n\times n$ matrices, each consisting of pairwise distinct variables. Then
\begin{equation}\label{eq:det}\Det(X\cdot Y) = \Det(X)\cdot \Det(Y)\end{equation}
has a polynomial-size $\PI$ proof. The proof is defined for $X=A, Y=B$ provided
$A[k], B[k]$ and $A[k]B[k]$ are invertible for every $k\in \{1,\dots, n\}$.
\end{enumerate}
\end{proposition}

\begin{proof}
Part \ref{part1} \iddosml{follows} from the definition of $\Det$. \iddosml{We omit the details.}\Commentdel{I checked Part (i), and the proof is not that direct (one has to go by induction). So I guess either we have to write the proof, or say that we ``omit the details" ?}

\ind Part \ref{part2} is proved by induction on $n$. If $n=1$, it is immediate. Assume that $n>1$.
Let
\[X=
\left( \begin{array}{l r} X_{1} & v_{1}^{t}\\
v_{2  } & x_{nn} \end{array}\right)\,, \,\,
~~~~Y=
\left( \begin{array}{l r} Y_{1} & u_{1}^{t}\\
u_{2} & y_{nn} \end{array}\right)\,.
\]

We want to construct a polynomial size proof of $\Det(XY)=\Det(X)\Det(Y)$ from the assumption $\Det(X_{1}Y_{1})=\Det(X_{1})\Det(Y_{1})$.
This implies that $\Det(XY)=\Det(X)\Det(Y)$ has a polynomial size proof.

By the definition of $ \Det $, we have
\begin{gather*}
\Det(X)=\Det(X_{1})\dd(X)\,,~~~~~ \Det(Y)= \Det(Y_{1})\dd(Y)\,~\hbox{\iddosml{ ~and} }\,\, \\ \Det(XY)= \Det(X_{1}Y_{1}+ v_1^tu_{2})\dd(XY)\,,
\end{gather*}

and we are supposed to prove:
\begin{equation}
\Det(X_{1}Y_{1}+ v_1^tu_{2})\dd(XY)= \Det(X_{1})\dd(X) \cdot \Det(Y_{1})\dd(Y)\,. \label{eq: sup}
\end{equation}
By the previous lemma, we have $\Det(X_{1}Y_{1}+ v_1^tu_{2})= \Det(X_{1}Y_{1})(1+u_{2}(X_{1}Y_{1})^{-1}v_{1}^{t})$. By the assumption $\Det(X_{1}Y_{1})=\Det(X_{1})\Det(Y_{1})$,  this yields
\[
\Det(X_{1}Y_{1}+ v_1^tu_{2})= \Det(X_{1})\Det(Y_{1})(1+u_{2}Y_{1}^{-1}X_{1}^{-1}v_{1}^{t})\,.
\]
Hence in order to prove (\ref{eq: sup}), it is sufficient to prove
\begin{equation*}
(1+u_{2}Y_{1}^{-1}X_{1}^{-1}v_{1}^{t})\dd(XY)= \dd(X)\dd(Y)\,.\end{equation*}
But this follows from (\ref{eq: ddXY}).

On the inductive step, we have assumed
 invertibility of $X,Y$, $XY$, $X_{1}$, $Y_{1}$ and $X_{1}Y_{1}$, as well as invertibility of  $X_{1}Y_{1}+v_1^tu_{2}$.
\iddoacc{The latter follows from the invertibility of $XY$, 
because $(X_{1}Y_{1}+v_1^tu_{2})^{-1} = ((XY)^{-1})[n-1]$ is used in the definition of $(XY)^{-1}$.}  
Since $X_{1}=X[n-1]$, $Y_{1}=Y[n-1]$, the proof altogether assumes invertibility of $X[k], Y[k]$ and $X[k]Y[k]$ for every $k\in \{1,\dots, n\}$.
\end{proof}
\QuadSpace 

Let us  explicitly state the important cases when the proof of $\Det(AB)=\Det(A)\Det(B)$ is defined. This is so, if $A$ and $B$ are invertible \iddoacc{and also at least one of the following conditions hold:} \mardel{Is this what you meant? The previous phrase was ambiguous. But here, it is still a bit confusing: if (i) holds then A and B are invertible in any case (so saying that BOTH (i) holds and A,B are invertible is redundant). So it would be best to write in brackets in (i) that "(in this case A,B are invertible automatically)".}
\begin{enumerate}
\item The entries of $A,B$ compute algebraically independent rational functions; 
\item $A$ is lower triangular or $B$ is upper triangular;
\item The entries of $A$ are field elements and the entries of $B$ are algebraically independent, or vice versa.
\end{enumerate}

The following lemma shows that elementary Gaussian operations are well-behaved with respect to $\Det$.

\begin{lemma} \label{lem:elop} Let $X=\{x_{ij}\}_{i,j\in[n]}$ be an $n\times n$  matrix of distinct variables. Then the following have polynomial-size $\PI$ proofs:
\begin{enumerate}
\item $\Det (X)=-\Det(X')$, where $X'$ is a matrix obtained from $X$ by interchanging two rows or columns.\label{exception}
\item $\Det(X'')= u\Det(X)$, where $X''$ is obtained by multiplying a row  in $X$ by $u$, such that $u^{-1}$ is defined (and similarly for a column).\label{part2}
\item $\Det(X)= \Det(X''')$, where $X'''$ is obtained by adding a row to a different row in $X$ (and similarly for columns). \label{part3}
\item $\Det(X)=x_{nn}\Det(X_{1}-x_{nn}^{-1}v_{1}^{t}v_{2}) $, where $X_{1},v_{1}$ and $v_{2}$ are from the decomposition (\ref{eq: X}). \label{partN}
\end{enumerate}
\end{lemma}

\begin{proof}
Parts \ref{part2} and \ref{part3} follow from Proposition \ref{lem:det} and the fact that $X''= A X$ and $X'''=A'X$, where $A,A'$ are suitable triangular matrices. 

\iddosml{For part \ref{exception}, we cannot directly infer it} from Proposition \ref{lem:det}, since $X'=TX$ implies only that $T$ is a transposition matrix and hence not invertible in our sense.  However, we can write $T= A_{1}A_{2}$, where $A_{1},A_{2}$ are invertible with $\Det(A_{1})\Det(A_{2})= -1$: note that
$\left(\begin{array}{l r} 0&1\\ 1 & 0\end{array}\right)=
\left( \begin{array}{l r} 1&1\\ 1 & 0\end{array}\right)\left(\begin{array}{r r} 1&0\\ -1 & 1\end{array}\right)$.
Since $X$ is a matrix of distinct variables, the following is defined:
\begin{align*}
\Det(A_{1}A_{2}X)=\Det (A_{1})\Det(A_{2}X)= \Det(A_{1})\Det(A_{2})\Det(X)\,.
\end{align*}
Part \ref{partN} follows from Lemma \ref{lem: vv}.
\end{proof}

\subsection{The determinant as a polynomial}\label{sec:det as polynomial}
Note that we cannot yet apply Theorem \ref{thm: main divisions} to obtain Theorem \ref{thm: main det}, because $\Det$ itself contains division gates. For our purpose it will suffice to compute the determinant by a circuit without division, denoted $\det(X)$, and construct a proof of  $\det(X)=\Det(X)$ in $\PI$. In order to do that, we will define $\det(X)$ as the $n$th term of the Taylor expansion of $\Det (I+zX)$ {at $ z=0 $}, as follows: using notation from Section \ref{sec:Taylor}, let
\begin{equation}\label{eq:new-det}
\det(X):= \Coef_{z^{n}}\left(\Det(I+zX)\right).
\end{equation}
Let us note that
\begin{enumerate}[(i)]
\item \iddosml{the circuit} $\det(X)$ indeed computes the determinant of $X$\iddosml{; and}\vspace{-6pt} 
\item \iddosml{the circuit} $\det(X)$ is a circuit without division\iddo{,} of syntactic degree $n$.
\end{enumerate}
This is because every variable from $X$ in the circuit $\Det(I+zX)$ occurs in a product with $z$\iddosml{, and thus} $\Coef_{z^{n}}(\Det(I+zX))$ is the  $n$th homogeneous part of the determinant of $I+X$, which is simply the determinant of $X$.
By the definition of $\Coef_{z^{n}}$, $\Coef_{z^{n}}(\Det(I+zX))$ contains exactly one inverse gate, namely the inverse of $\Den(\Det(I+zX))$ at the point $z=0$. But $a:= (\Den(\Det(I+zX)))(z/0)^{\sharp}$ is a constant circuit computing a non-zero field element, and we can identify $a^{-1}$ with the field constant it computes.

\begin{lemma}\label{lem: CH} Let $X$ be an {$ n\times n $} matrix of distinct variables.
There exist circuits with divisions $P_{0},\dots, P_{n-1}$ not containing the variable $z$, such that
        \[\Det(zI_{n}+X)= z^{n}+ P_{n-1}z^{n-1}+\dots+P_{0}
        \]
has a polynomial-size $\PI(\F)$ proof. Moreover, this proof is defined for $z=0$.
\end{lemma}

\begin{proof}
{Let $ F $ be a circuit in which $z$ does not occur in the scope of any inverse gate. Then, we define the \emph{$z$-degree of $F$}
as the syntactic-degree of $ F $ considered as a circuit computing a univariate polynomial in $ z $ (so that all other variables are treated as constants).}

By induction, we will construct matrices $A_{1},\dots, A_{n}$  with the following properties:
\begin{enumerate}[1.]
\item $A_{1}= X+zI_n$,
\item Every $A_{k}$ is an $(n-k+1)\times (n-k+1)$ matrix of the form
\[\left(\begin{array}{l r} z^{k}+f & w\\
v^{t} & zI_{n-k}+Q
\end{array}\right)\,
\]
where all the entries are circuits with division\iddosml{s} in which $z$ does not occur in the scope of any division gate, $v,w$ are $1\times (n-k)$ vectors and moreover:
 $f$ as well as every entry of $w$ have $z$-degree less than $k$ and $v,Q$ do not contain the variable $z$.
 \label{item:c3}

\item The identity $\Det(A_{k})=\Det(A_{k+1})$ has a polynomial-size proof. 
\item The entries of $A_{k}$ are algebraically independent (this is to guarantee that divisions are defined).
\end{enumerate}
\FullSpace

Assume that $A_{k}$  is given, and let us partition it as
\[A_{k} = \left(\begin{array}{l l l} z^{k}+f_{1} & w & f_{2} \\
                                        u_{1}^{t} & zI_{m}+ Q & u_{2}^{t} \\
                                        a_{1} & v & z+a_{2}
                                         \end{array}\right)\,
                                         \]
where  $m= (n-k-1)$ and we allow the \iddosml{possibility} that $m=0$. \iddosml{By assumption $f_{1}, w$ and $ f_{2}$  have $z$-degree smaller than $ k$, and $ z $ does not occur in $u_1,u_2,Q, a_1, a_2$ and $v$.}
By \iddosml{Lemma \ref{lem:elop} part \ref{exception}}, we can switch the first and last column to obtain a $\PI$ proof of
\[\Det(A_{k}) = - \Det
\left(\begin{array}{l l l} f_{2}&w& z^{k}+f_{1} \\
                                         u_{2}^{t}& zI_{m}+ Q & u_{1}^{t}  \\
                                           z+a_{2}& v& a_{1}
                                         \end{array}\right)\,. \]
By Lemma \ref{lem:elop} part \ref{partN}, we have
\begin{align*}
\Det(A_{k}) = - a_{1} \Det
\left(\begin{array}{l r} f_{2} - a_{1}^{-1}(z^{k}+f_{1})(z+a_{2})&w-a_{1}^{-1}(z^{k}+f_{1})v  \\
                                         u_{2}^{t}- a_{1}^{-1}u_{1}^{t}(z+a_{2})&zI_{m}+ Q-a_{1}^{-1}u_{1}^{t}v   \\             
                                         \end{array}\right)\,
                                         =\\
                                \Det
\left(\begin{array}{l r}  (z^{k}+f_{1})(z+a_{2})- a_{1}f_{2}&a_{1}w-(z^{k}+f_{1})v  \\
                                         u_{2}^{t}- a_{1}^{-1}u_{1}^{t}(z+a_{2})&zI_{m}+ Q-a_{1}^{-1}u_{1}^{t}v   \\                                     
                                         \end{array}\right)
                                         \,.
        \end{align*}
We can write $(z^{k}+f)(z+a_{2})= z^{k+1}+ (fz+a_{2}z^{k}+fa_{2})$, where the $z$-degree of $(fz+a_{2}z^{k}+fa_{2})$ as well as of every entry of $a_{1}w-(z^{k}+f_{1})v $ is at most $k$. Hence the matrix is of the \iddosml{correct} form, apart from the occurrence of $zu_{1}^{t}$ in the first column. This can be remedied by multiplying by $\left(\begin{array}{l r} 1 & 0\\ -a_{1}^{-1}u_{1}^{t} & I_m
\end{array}\right)$ from the right to obtain $A_{k+1}$ of the required form.

This indicates that, given a circuit computing $A_{k}$, we can compute $A_{k+1}$ using polynomially many additional gates. Altogether, every  $A_{k}$ has a polynomial size circuit. The proof of $\Det (A_{k})=\Det (A_{k+1})$ has a polynomial number of lines and, as it involves polynomial size circuits, also polynomial size.

Finally, we obtain a polynomial size proof of $\Det (A_{n})=\Det( A_{1})= z^{n}+ f$, where $f$ is a circuit with $ z $-degree \iddosml{smaller than} $n $ in which $ z $ is not in the scope of any division gate. Writing $f$ as  $\sum_{i=0}^{n-1}P_{i}z^{i}$ concludes the lemma.
\end{proof}

\begin{proposition}\label{lem:charpoly}
\
\begin{enumerate}
\item \label{char-jedna} If $U$ is a triangular matrix with $u_{1},\dots, u_{n}$ on the diagonal then $\det(U)=u_{1}\cdots u_{n}$ has a polynomial size $\PI$ proof.

\item \label{char-dva} Let $X$ be an {$ n\times n $} matrix of distinct variables. Then
\[\Det(X)=\det(X)\,\]
has a polynomial-size $\PI$ proof.
\end{enumerate}
\end{proposition}

\begin{proof}
Part \ref{char-jedna} follows from Proposition \ref{lem:det}. For we have $\Det(I_{n}+zU)=(1+zu_{1})\cdots (1+zu_{n})$, and the proof is defined for $z=0$. \Commentdel{Perhaps you need here to state that: ``(we assume that $z$ does not occur in $U$)".\  Otherwise, you can have $1+z\cd\frac{-1}{z}$ which will make the proof undefined.}Thus, by Proposition \ref{prop:Taylor} 
$$
\det(U)=\Coef_{z^{n}}((1+zu_{1})\cdots (1+zu_{n}))= u_{1}\cdots u_{n}
$$
has a polynomial\iddosml{-}size $\PI$ proof. \Commentdel{Maybe we should write that this holds, since there is a short proof of $ (1+zu_{1})\cdots (1+zu_{n})=z^n\cd u_1\cdots u_n +T$, for some $T$ with $ z $-degree $<n$.}

\iddosml{Part} \ref{char-dva} follows from the previous lemma, \iddosml{ as follows}. We obtain polynomial-size $ \PI $ proof\iddosml{s} of the following substitution instance:\mardel{Need to state why is this substitution instance defined?}
\begin{equation}\label{eq:some substitution instance}
        \Det(zI_{n}+X^{-1})= z^{n}+ Q_{n-1}z^{n-1}+\dots+Q_{0},
\end{equation}
where the $Q_i $'s are circuits with divisions \iddosml{that} do not contain the variable $ z $ and the proof is defined for $z=0$.

By Proposition \ref{lem:det} we have a polynomial-size $ \PI$  proof of
\[
        \Det(I_{n}+{zX})= \Det(zI_{n}+X^{-1})\cd \Det (X)\,.
        \]
The proof is defined for $z=0$ (as is witnessed by letting $X:=I_{n}$).\Commentdel{I don't understand this: I think that the proof is defined for z=0 already due to the conditions stated in Proposition 35, namely, because  $X[k]$ and $(0\cd I_n+X^{-1})[k]$ are invertible for every k=1,...,n }
From equation (\ref{eq:some substitution instance}) we get a polynomial-size proof of
\begin{equation*}\label{eq:Detz}\Det(I_{n}+zX)=  z^{n}\Det (X) +  z^{n-1}Q'_{n-1}+\dots +Q'_{0}, \end{equation*}
where $Q'_{n-1},\dots, Q'_{0}$  do not contain $z$. The proof is defined for $z=0$ and so Proposition \ref{prop:Taylor} gives a polynomial-size $ \PI $ proof of
\begin{equation*}
        \Coef_{z^{n}}(\Det(I_{n}+zX))=
\Coef_{z^{n}}(z^{n}\Det (X) +  z^{n-1}Q'_{n-1}+\dots +Q'_{0}).
\end{equation*}
But by the definition of $ \det(X) $,  $\Coef_{z^{n}}(\Det(I+zX)) $ is $ \det (X)$ and by the definition of $ \Coef_{z^n} $, $\Coef_{z^{n}}(z^{n}\Det (X) +  z^{n-1}Q'_{n-1}+\dots +Q'_{0}) $ is $  \Det (X)$, and we are done.
\end{proof}
\FullSpace

\section{Concluding the main theorem}
We can now finally prove Theorem \ref{thm: main det} (Main Theorem), which we rephrase as follows:

\begin{proposition}[Theorem \ref{thm: main det}, rephrased]\label{prop: detc}
Let $X,Y,Z$ be $n\times n$ matrices such that $X,Y$ consist of different variables and $Z$ is a triangular matrix with $z_{11},\dots, z_{nn}$ on the diagonal. Then there exist an arithmetic circuit $\detc$ and a formula $\det_f$ such that:
\begin{enumerate}
\item The identity \label{prop - jedna} $\detc(XY)=\detc(X)\cd\detc(Y)$ and $\detc(Z)= z_{11}\cdots z_{nn}$ have polynomial-size $ O(\log^2 n)$\iddo{-}depth proofs in $\PC$.
\item\label{prop - dva}
The identity $\det_f(XY)=\det_f(X)\cd\det_f(Y)$ and $\det_f(Z)= z_{11}\cdots z_{nn}$ have $\PF$ proofs of size $n^{O(\log n)}$.
\end{enumerate}
\end{proposition}

\begin{proof}
Let $\det(X)=\Delta_{z^n}\Det(I+zX)$ be the circuit defined in (\ref{eq:new-det}).
Lemma \ref{lem:charpoly} part \ref{char-dva} and Proposition \ref{lem:det} imply that the equations
\begin{equation}
\det(XY)=\det(X)\cd\det(Y)\ \ \ \ \ \mbox{and}\ \ \ \ \ \det(Z)= z_{11}\cdots z_{nn} \label{eq: detXY}
\end{equation}
have polynomial-size $\PI$ proofs. By definition, the syntactic degree of $\det(X)$ is at most $ n$.
Hence, by Theorem \ref{thm: main divisions} the identities in (\ref{eq: detXY}) have polynomial-size $\PC$ proofs.
This almost concludes part \ref{prop - jedna}, except for the bound on the depth. To bound the depth, let
\[\det\nolimits_c(X):= [\det(X)], \]
where $[F]$ is the balancing operator as defined in Section \ref{sec:balancing proofs}. Thus, Theorem \ref{thm: main balancing} implies that
\[
         \tr{\det(XY)}=\tr{\det(X)\cd\det(Y)}\qquad\mbox{and}\qquad\tr{\det(Z)}= \tr{z_{11}\cdots z_{nn}}
\]
have $\PC$ proofs of polynomial-size and depth $O(\log^2 n)$. By means of Lemma \ref{lem:simulation}, we have such proofs also for
\[
        \tr{\det(X)\cd\det(Y)}=\tr{\det(X)}\cd\tr{\det(Y)}=
                \det\nolimits_c(X)\cd\det\nolimits_c(Y) \qquad\mbox{and}\qquad \tr{\det(Z)}= z_{11}\cdots z_{nn}.
\]
Hence it is sufficient to construct (polynomial-size and $O(\log^2 n)$ depth proofs) of
\[\tr{\det(XY)}= \det\nolimits_c(XY) \qquad\mbox{and}\qquad \tr{\det(Z)}= \det\nolimits_c(Z)\]
{(note that defining $ \detc(X) $ as $ \tr{\det(X)} $ does not imply that $ \tr{\det(XY)}= \det\nolimits_c(XY) $)}. This follows from the following more general claim:

\begin{claim*} Let $F(x_1/g_1,\dots, x_n/g_n )$ be a circuit of size $s$ and syntactic degree $d$. Then
\[\tr{F(x_1/g_1,\dots, x_n/g_n )}= \tr{F(x_1,\dots, x_n)}(x_1/\tr{g_1},\dots, x_n/\tr{g_n})\]
has a $\PC$ proof of size $\poly(n,d)$ and depth $O(\log d\log s+\log^2d)$.
\end{claim*}
\begin{proof}
This follows by induction using Lemma \ref{lem:simulation}. We omit the details.
\end{proof}
\medskip
To prove part \ref{prop - dva}, recall the definition  of $F^\bullet$ from Remark \ref{def:similar}. Let {$\det_f(X):=(\detc(X))^\bullet$}. Then the statement follows from part \ref{prop - jedna}
 and Claim 1 \iddosml{in} the proof of Theorem \ref{thm:5-restated-is now flying}.
\end{proof}

We should note that in the $\PC$-proof of the equation $\det(XY)=\det(X)\cd\det(Y)$ no divisions occur and so it is defined for any substitution. In particular,
\[\det(AX)= \det(A)\cd\det(X)= a\iddo{\cd}\det(X)\]
has a short $\PC $ proof for any matrix $A$ of field elements whose determinant is $a\in \F$.
Similarly, the elementary Gaussian operations stated in Lemma \ref{lem:elop} carry over to polynomial-size $\PC$ proofs of the corresponding properties of $\det$.

\section{Applications}\label{sec:applications}

In this section, we prove Propositions \ref{prop:Valiant} and  \ref{prop:XY}, as well as a $\PC$-version of Cayley-Hamilton theorem.
First, one should show that the cofactor expansion of the determinant has short proofs. For an $n\times n$ matrix $X$ and $ i,j \in [n]$, let $X_{i,j}$ denote the $(n-1)\times (n-1)$-matrix obtained by removing the $i$th row and $j$th column from $X$. Let $\Adj(X)$ be the $n\times n$ matrix whose $(i,j)$-th entry is $(-1)^{i+j}\detc(X_{j,i})$ (where $\detc$ is the circuit from Proposition \ref{prop: detc}).

\begin{proposition}[Cofactor expansion]\label{prop:cofactor}
Let $X=\{x_{ij}\}_{i,j\in [n]}$ be an $n\times n$ matrix, {for variables $ x_{ij} $}. Then the following identities have polynomial-size $ O(\log^2 n) $-depth $\PC$ proofs:
\begin{enumerate}[(i)]
\item $\detc(X)= \sum_{j=1}^n (-1)^{i+j}x_{ij}{\detc}(X_{i,j})$,\, for any $i\in [n] $;\label{cof1}
\item \label{cor:adj} $X\cdot \Adj(X)=\Adj(X)\cdot X= \detc(X)\cdot I$.\label{cof2}
\end{enumerate}
\end{proposition}

\begin{proof}
For part \iddosml{(\ref{cof1})} we prove
 ${\det}_{c}(X)= \sum_{j=1}^n (-1)^{1+j}x_{1j}{{\det}_{c}}(X_{1,j})$. The general case follows if we multiply $X$ by an appropriate permutation matrix using Proposition \ref{prop: detc}.
It is sufficient to construct a polynomial size $\PI$ proof,  for we can then eliminate the division gates by means of Theorem \ref{thm: main divisions} and bound the depth of the proof by means of Theorem \ref{thm: main balancing}.

For $j\in \{1,\dots,n\}$, let $X_{j}$ be the matrix obtained by replacing $x_{1i}$ by $0$ in $X$, for every $i\not=j$.
We want to show that
\begin{eqnarray} \detc(X)&=&\detc(X_{1})+\dots +\detc(X_{n}) \label{cof3}\\
                                                   \detc(X_{j})&=& (-1)^{1+j}x_{1j}\detc(X_{1,j})\,, \,\, j\in \{1,\dots,n\} \label{cof4}
\end{eqnarray}
have polynomial size $\PI$ proofs.

For (\ref{cof4}), it is sufficient to consider $j=1$, the other cases follow by an permutation of rows.
By Proposition \ref{lem: LU} we there exist a lower resp. upper triangular matrix $L$ and $U$ such that $X_{1,1}=LU$ has a polynomial size proof.
If $w:=(x_{21},\dots, x_{(n-1)1})$, we have
\[X_{1}= \left(\begin{array}{l r} x_{11}& 0\\ w^{t} & X_{1,1}\end{array}\right)=\left(\begin{array}{l r} x_{11}& 0\\ w^{t}& L\end{array}\right)\left(\begin{array}{l r} 1& 0\\ 0 &U\end{array}\right)\, \]
and so by Proposition \ref{prop: detc}
\[\detc(X_{1})= x_{11}\detc(L)\detc(U)= x_{11}\detc(LU)=x_{11}\detc(X_{1,1})\,.\]

\iddosml{Equation} (\ref{cof3}) follows from the  general identity
\[\detc (X[u+v])=\detc(X[u])+\detc(X[v])\,,\] where $X[v]$ denotes the matrix obtained by replacing the first row of $X$ by the vector $v$.
Writing $u=(u_{1},\bar u)$ and  $v=(v_{1},\bar v)$, we have
\[ X[u]=
\left(\begin{array}{l r}u_{1}& \bar u\\ w^{t} & X_{1,1}\end{array}\right)= \left(\begin{array}{l r}u_{1}-\bar u X_{1,1}^{-1}w^{t}& \bar u X_{1,1}^{-1}\\ 0 & I_{n}\end{array}\right) \cdot \left(\begin{array}{l r}1& 0 \\ w^t &X_{1,1} \end{array}\right)\,.
\]
Hence
$X[u]=\detc(X_{1,1})(u_{1}-\bar u X_{1,1}^{-1}w^{t})$,
and similarly for $X[v]$ and $X[u+v]$.
Therefore
\begin{align*}
\detc(X[u+v])&=\detc(X_{1,1})(u_{1}+v_{1}-(\bar u+\bar v) X_{1,1}^{-1}w^{t})
\\& = \detc(X_{1,1})(u_{1}-\bar u X_{1,1}^{-1}w^{t})+ \detc(X_{1,1})(v_{1}-\bar v X_{1,1}^{-1}w^{t})\\
& = \detc(X[u])+\detc(X[v])\,.
\end{align*}

Part \iddosml{(\ref{cof2})} is an application of \iddosml{part (\ref{cof1})}. The $i,j$-entry of $X\cdot \Adj(X)$
is \[a_{ij}= \sum_{k=1}^{n}(-1)^{i+k}x_{ik}\detc(X_{j,k})\,.\]
Hence we already know that $a_{ij}=\detc(X)$ whenever $i=j$ and it remains to show that $a_{ij}=0$ if $i\not=j$. By part \ref{cof2}
$\sum_{k=1}^{n}(-1)^{i+k}x_{ik}\detc(X_{j,k})=\detc(Y)$, where $Y$ is the matrix obtained by replacing the $j$-th row in $X$ by $(x_{i1},\dots, x_{in})$.
I.e., if $i\not =j$, $Y$ contains two identical rows. Then  $Y$ can be written as $Y= AJY$, where $J$ is a diagonal matrix with some entry on the diagonal equal to zero, and so $\detc(Y)= \detc(A)\detc(J)\detc(Y)= 0$. The proof for $\Adj(X)X$ is similar, or note that we can now conclude $\Adj(X)= \detc(X)X^{-1}$.\mardel{Why did you write here "\underline{or}, note that.."?}

\end{proof}

\begin{proposition}[Proposition \ref{prop:XY} restated]
The identities $ YX=I_{n}$ have polynomial-size and $ O(\log^2 n)$-depth $ \PC$ proofs from the equations $XY=I_{n}$. In the case of\, $\PF$, the proofs have quasipolynomial-size.
\end{proposition}

\begin{proof} Note that we are dealing with a $\PC$ proof from assumptions, and hence we are not allowed to use division gates. The proof is constructed as follows. Assume  $XY=I_{n}$. By Proposition \ref{prop: detc}, this gives $\detc(X)\detc(Y)=1$. By Proposition \ref{prop:cofactor}, we can multiply from left both sides of $XY=I_{n}$ by $\Adj(X)$, to obtain $\detc(X) Y= \Adj(X)$. Hence,
\[{\det}_{c}(X) Y X= \Adj(X) X= {\det}_{c}(X)I_{n} ,\]
and so
\[{\det}_{c}(Y) {\det}_{c}(X) Y X= \detc(Y) \detc(X) I_{n} ,\]
which, using $ \detc(X)\detc(Y)=1$ gives $YX=I_{n}$.
The $\PF$ proof is identical, except that the steps involving the determinant require a quasipolynomial size.
\end{proof}
\FullSpace
\begin{proof}
[Proof of Proposition \ref{prop:Valiant}]
{The proof proceeds via a simulation of the construction in \cite{Val79:ComplClass}} (compare also with the presentation in \cite{HWY10:CCC}). The matrix $M$ is constructed inductively with respect to the size of the formula. It is convenient to maintain the property
\begin{equation*} \label{eq:weakly}M_{i,i+1}= 1 \qquad \hbox{and}\qquad M_{i,j}=0, \ \hbox{if}\  j>i+1\,. \end{equation*}

Let us call matrices of this form  \emph{nearly triangular}. Let $M_1, M_2$ be nearly triangular matrices of dimensions $s_1\times s_1$ and $s_2\times s_2$, respectively. In order to prove the correctness of the simulation of Valiant's construction \cite{Val79:ComplClass},
it is sufficient to show that the following equations have polynomial-size $\PC$ proofs:

\begin{enumerate}
\item $\detc(M)=\detc(M_1)\cdot\detc (M_2)$, where
$$M = \left( \begin{array}{cc}
M_1 & E \\
0 & M_2 \\
\end{array} \right) \,, $$ and $E$ has $1$ in the lower left corner and $0$ otherwise.
\item $\detc(M)= \detc(M_1)+\detc(M_2)$, with
$$M = \left( \begin{array}{cccc}
1 &  v & 0 & 0 \\
0 & M_1 & v_1 & 0  \\
M_2[1] & 0 & v_2 & M_2[2^{+}] ,
\end{array} \right) , $$
where $v$ is a row vector with $1$ in the leftmost entry and $0$ elsewhere,
$v_1$ is a column vector with $1$ in the bottom entry and $0$ elsewhere, $v_2$ is a column vector with $(-1)^{s_2+1}$ in the bottom entry and $0$ elsewhere,
$M_2[1]$ is the first column of $M_2$,
and $M_2[2^{+}]$ is the matrix $M_2$ without the first column.
\end{enumerate}
Both parts are an application of Proposition \ref{prop:cofactor}.
\end{proof}
\mardel{I couldn't see why M has only variables and field elements as entries.}

\new{\subsection*{Cayley-Hamilton theorem}}

Let $X=\{x_{i,j}\}_{i,j\in [n]}$ be an $n\times n$ matrix of distinct variables. For $i\in \{0,\dots, n\}$, let $p_{i}$ be the circuit in variables $X$
defined by
\[p_{i}:= \Coef_{z^{i}}(\detc(zI_{n}-X))\,\]
and let $P_{X}(z)$ be the circuit
\[P_{X}(z):= \sum_{i=0}^{n}p_{i}z^{i}\,.\]
 $P_{X}(z)$ computes the characteristic polynomial of the matrix $X$ and we can prove the following version of Cayley-Hamilton theorem:

\begin{proposition} \label{prop: CH} \[P_{X}(X)= \sum_{i=0}^{n}p_{i}X^{i}=0\] has a polynomial-size $\PC$-proof.
\end{proposition}

As \iddosml{before}, if we replace the $p_{i}$'s by their balanced versions, we can obtain a polynomial-size $\PC$-proof of depth $O(\log^{2}(n))$.

\begin{proof}
Since\, $\detc(zI_{n}-X)$ \,has a syntactic degree $n$,  we have a polynomial\iddo{-}size proof of
$\detc(zI_{n}-X)= P_{X}(z)$ by Proposition \ref{prop:Taylor}.
Proposition \ref{prop:cofactor} gives
\[\Adj(zI_{n}-X)\cdot (zI_{n}- X)= \detc(zI_{n}-X)I_{n}= P_{X}(z)I_{n}\,.\]
Since every entry of $\Adj$ has a syntactic degree \iddosml{less than} $ n$, we can write
$\Adj(zI_{n}-X)= \sum_{i=0}^{n-1}A_{i}z^{i}$,
where the matrices $A_{i}$ do not contain $z$. Hence we also have
\[
        \left(\sum_{i=0}^{n-1}A_{i}z^{i}\right)\cdot (zI_{n}-X)=P_{X}(z)I_{n}\,.
\]
Expanding the left-hand side and collecting terms with the same power of $z$ gives
\begin{equation}
-A_{0}X+ \sum_{i=1}^{n-1}(A_{i-1}-A_{i}X)z^{i}+A_{n-1}z^{n}= p_{X}(z)I_{n}\,.\label{eq: CH}
\end{equation}
Since $P_{X}(z)=\sum_{i=0}^{n}p_{i}z^{i}$, where the $p_{i}$'s do not contain $z$, we
can compare the coefficients on the left and right-hand side of (\ref{eq: CH}) (see Proposition \ref{prop:Taylor}) to conclude
\[
p_{0}I_{n}= -A_{0}X\,,~~~~~ p_{i}I_{n}= A_{i-1}-A_{i}X\,\,\hbox{ ~~if~ }i\in \{1,\dots, n-1\}\,,~~~\,\, p_{n}I_{n}= A_{n-1}\,. 
\]
Hence 
\begin{align*}
\sum_{i=0}^{n}p_{i}X^{i} &= p_{0}I_{n}+p_{1}X+p_{2}X^{2}\iddosml{+}\dots + p_{n-1}X^{n-1}+ p_{n}X^{n}\\
&=-A_{0}X+ (A_{0}-A_{1}X)X+(A_{1}-A_{2}X)X^{2}\iddosml{+}\dots +(A_{n-2}-A_{n-1}X)X^{n-1}+A_{n-1}X^{n}\\ &=(-A_{0}X+A_{0}X)+ (-A_{1}X^{2}-A_{1}X^{2})+\dots\iddosml{+} (-A_{n-1}X^{n}+A_{n-1}X^{n})\\
&=0\,.
\end{align*}
\end{proof}

\bibliographystyle{alpha}
\bibliography{PrfCmplx-Bakoma}

\begin{thebibliography}{VSBR83}

\bibitem[BBP95]{BBP95}
Maria~Luisa Bonet, Samuel~R. Buss, and Toniann Pitassi.
\newblock Are there hard examples for {F}rege systems?
\newblock In {\em Feasible mathematics, II (Ithaca, NY, 1992)}, volume~13 of
  {\em Progr. Comput. Sci. Appl. Logic}, pages 30--56. Birkh\"auser Boston,
  Boston, MA, 1995.

\bibitem[Ber84]{Ber84}
Stuart~J. Berkowitz.
\newblock On computing the determinant in small parallel time using a small
  number of processors.
\newblock {\em Inf. Process. Lett.}, 18:147--150, 1984.

\bibitem[BP98]{BP98}
Paul Beame and Toniann Pitassi.
\newblock Propositional proof complexity: past, present, and future.
\newblock {\em Bull. Eur. Assoc. Theor. Comput. Sci. EATCS}, (65):66--89, 1998.

\bibitem[HT09]{HT08}
Pavel Hrube\v{s} and Iddo Tzameret.
\newblock The proof complexity of polynomial identities.
\newblock In {\em Proceedings of the 24th IEEE Conference on Computational
  Complexity (CCC)}, pages 41--51, 2009.

\bibitem[HWY10]{HWY10:CCC}
Pavel Hrube\v{s}, Avi Wigderson, and Amir Yehudayoff.
\newblock Relationless completeness and separations.
\newblock In {\em Proceedings of the 25th IEEE Conference on Computational
  Complexity}, pages 280--290, 2010.

\bibitem[Hya79]{Hya79}
Laurent Hyafil.
\newblock On the parallel evaluation of multivariate polynomials.
\newblock {\em SIAM J. Comput.}, 8(2):120--123, 1979.

\bibitem[Je{\v{r}}04]{Jer04}
Emil Je{\v{r}}{\'a}bek.
\newblock Dual weak pigeonhole principle, {B}oolean complexity, and
  derandomization.
\newblock {\em Ann. Pure Appl. Logic}, 129(1-3):1--37, 2004.

\bibitem[Kra95]{Kra95}
Jan Kraj{\'{\i}}{\v{c}}ek.
\newblock {\em Bounded arithmetic, propositional logic, and complexity theory},
  volume~60 of {\em Encyclopedia of Mathematics and its Applications}.
\newblock Cambridge University Press, Cambridge, 1995.

\bibitem[RY08]{RY08-balancing}
Ran Raz and Amir Yehudayoff.
\newblock Balancing syntactically multilinear arithmetic circuits.
\newblock {\em Computational Complexity}, 17:515--535, 2008.

\bibitem[SC04]{SC04}
Michael Soltys and Stephen Cook.
\newblock The proof complexity of linear algebra.
\newblock {\em Ann. Pure Appl. Logic}, 130(1-3):277--323, 2004.

\bibitem[Sch80]{Sch80}
Jacob~T. Schwartz.
\newblock Fast probabilistic algorithms for verification of polynomial
  identities.
\newblock {\em Journal of the ACM}, 27(4):701--717, 1980.

\bibitem[Seg07]{Seg_BSL07}
Nathan Segerlind.
\newblock The complexity of propositional proofs.
\newblock {\em Bull. Symbolic Logic}, 13(4):417--481, 2007.

\bibitem[Sol01]{Sol_PhD}
Michael Soltys.
\newblock {\em The complexity of derivations of matrix identities}.
\newblock PhD thesis, University of Toronto, Toronto, Canada, 2001.

\bibitem[Sol05]{Sol05}
Michael Soltys.
\newblock Feasible proofs of matrix properties with csanky's algorithm.
\newblock In {\em 19th International Workshop on Computer Science Logic}, pages
  493--508, 2005.

\bibitem[Str73]{Str73}
Volker Strassen.
\newblock Vermeidung von divisionen.
\newblock {\em J. Reine Angew. Math.}, 264:182--202, 1973.
\newblock (in German).

\bibitem[SU04]{SU04}
Michael Soltys and Alasdair Urquhart.
\newblock Matrix identities and the pigeonhole principle.
\newblock {\em Arch. Math. Logic}, 43(3):351--357, 2004.

\bibitem[SY10]{SY10}
Amir Shpilka and Amir Yehudayoff.
\newblock Arithmetic circuits: A survey of recent results and open questions.
\newblock {\em Foundations and Trends in Theoretical Computer Science},
  5(3-4):207--388, 2010.

\bibitem[Val79]{Val79:ComplClass}
Leslie~G. Valiant.
\newblock Completeness classes in algebra.
\newblock In {\em Proceedings of the 11th Annual ACM Symposium on the Theory of
  Computing}, pages 249--261. ACM, 1979.

\bibitem[VSBR83]{VSB+83}
Leslie~G. Valiant, Sven Skyum, S.~Berkowitz, and Charles Rackoff.
\newblock Fast parallel computation of polynomials using few processors.
\newblock {\em SIAM J. Comput.}, 12(4):641--644, 1983.

\bibitem[Zip79]{Zip79}
Richard Zippel.
\newblock Probabilistic algorithms for sparse polynomials.
\newblock In {\em Proceedings of the International Symposiumon on Symbolic and
  Algebraic Computation}, pages 216--226. Springer-Verlag, 1979.

\end{thebibliography}

\end{document}